\documentclass[11pt]{article}

\usepackage{amsfonts,amsmath,amssymb,amsthm}
\usepackage{algorithm}
\usepackage[noend]{algorithmic}
\usepackage{float}
\usepackage{color}
\usepackage{graphicx} 
\usepackage{enumerate}
\usepackage{enumitem}

\newtheorem{theorem}{Theorem}[section]
\newtheorem{proposition}[theorem]{Proposition}
\newtheorem{observation}[theorem]{Observation}
\newtheorem{lemma}[theorem]{Lemma}
\newtheorem{corollary}[theorem]{Corollary}

\newtheorem{pseudo-theorem}[theorem]{Pseudo-Theorem}

\newcommand{\F}{\mathcal{F}}
\newcommand{\C}{\mathcal{C}}
\newcommand{\I}{\mathcal{I}}

\newcommand{\ct}{C_{\text{tight}}}
\newcommand{\cs}{C_{\text{slack}}}
\renewcommand{\S}{\mathcal{S}}

\renewcommand{\P}{\mathcal{P}}
\newcommand{\bE}{\ensuremath{\mathbf{E}}}

\newcommand{\qf}{Q_{\text{f}}}
\newcommand{\qp}{Q_{\text{p}}}

\usepackage[margin=1in]{geometry}
\begin{document}

\title{Approximation algorithms for stochastic clustering\thanks{Research supported in part by NSF Awards CNS-1010789, CCF-1422569, CCF-1749864, CCF-1566356, CCF-1717134, a gift from Google, Inc., and by research awards from Adobe, Inc.\ and Amazon, Inc.}}

\author{  
David G. Harris\thanks{Department of Computer Science, University of Maryland,
College Park, MD 20742. Email:
\texttt{davidgharris29@gmail.com}}
\and
Shi Li\thanks{University at Buffalo, Buffalo, NY. Email: \texttt{shil@buffalo.edu}}
\and
Thomas Pensyl\thanks{Bandwidth, Inc. Raleigh, NC. Email: \texttt{tpensyl@bandwidth.com}}
\and
Aravind Srinivasan\thanks{Department of Computer Science and Institute for Advanced Computer Studies, University of Maryland, College Park, MD 20742.  Email: 
\texttt{srin@cs.umd.edu}} 
\and Khoa Trinh\thanks{Google, Mountain View, CA 94043. Email: \texttt{khoatrinh@google.com}}}

\date{}

\maketitle

\begin{abstract}
We consider stochastic settings for clustering, and develop provably-good approximation algorithms for a number of these notions. These algorithms yield better approximation ratios compared to the usual deterministic clustering setting. Additionally, they offer a number of advantages including clustering which is fairer and has better long-term behavior for each user. In particular, they ensure that \emph{every user} is guaranteed to get good service (on average). We also complement some of these with impossibility results. 

\end{abstract}

This is an extended version of a paper that appeared in the \emph{Conference on Neural Information Processing Systems} (NeurIPS) 2018.

\section{Introduction}
\emph{Clustering} is a fundamental problem in machine learning and data science. A general clustering task is to partition the given datapoints such that points inside the same cluster are ``similar'' to each other.  More formally, consider a set of datapoints $\C$ and a set of ``potential cluster centers'' $\F$, with a metric $d$ on $\C \cup \F$. We define $n := |\C \cup \F|$.  Given any set $\S \subseteq \F$, each $j \in \C$ is associated with the key statistic $d(j, \S) = \min_{i \in \S} d(i, j)$. The typical clustering task is to select a set $\S \subseteq \F$ which has a small size and which minimizes the values of $d(j, \S)$.

The size of the set $\S$ is often fixed to a value $k$, and we typically ``boil down'' the large collection of values $d(j, \S)$ into a single overall objective function. A variety of objective functions and assumptions on sets $\C$ and $\F$ are used. The most popular problems include\footnote{In the original version of the $k$-means problem, $\C$ is a subset of $\mathbb{R}^\ell$ and $\F = \mathbb{R}^\ell$ and $d$ is the Euclidean metric. By standard discretization techniques (see, e.g., \cite{matousek, feldman}), the set $\F$ can be reduced to a polynomially-bounded set with only a small loss in the value of the objective function.}
\begin{itemize}
	\item the $k$-center problem: minimize the value $\max_{j \in \C} d(j, \S)$ given that $\F = \C$.
	\item the $k$-supplier problem: minimize the value $\max_{j \in \C} d(j, \S)$ (where $\F$ and $\C$ may be unrelated); 
	\item the $k$-median problem: minimize the summed value $\sum_{j \in \C} d(j, \S)$; and 
	
	\item the $k$-means problem: minimize the summed square value $\sum_{j \in \C} d(j, \S)^2$.	
\end{itemize}

An important special case is when $\C = \F$ (e.g. the $k$-center problem); since this often occurs in the context of data clustering, we refer to this as the \textbf{self-contained clustering (SCC)} setting.

These classic NP-hard problems have been studied intensively for the past few decades. There is an alternative interpretation from the viewpoint of operations research:  the sets $\F$ and $\C$ can be thought of as ``facilities'' and ``clients'', respectively. We say that $i \in \F$ is \emph{open} if $i$ is placed into the solution set $\S$. For a set $\S \subseteq \F$ of open facilities, $d(j, \S)$ can then be interpreted as the connection cost of client $j$. This terminology has historically been used for clustering problems, and we adopt it throughout for consistency. However, our focus is on the case in which $\C$ and $\F$ are arbitrary abstract sets in the data-clustering setting.

Since these problems are NP-hard, much effort has been paid on algorithms with ``small'' provable \emph{approximation ratios/guarantees}: i.e., polynomial-time algorithms that produce solutions of cost at most $\alpha$ times the optimal. 
The current best-known approximation ratio for $k$-median is $2.675$ by Byrka et.\ al.\ \cite{byrka} and it is NP-hard to approximate this problem to within a factor of $1+2/e \approx 1.735$ \cite{jms}. The recent breakthrough by Ahmadian et.\ al.\ \cite{svensson} gives a $6.357$-approximation algorithm for $k$-means, improving on the previous approximation guarantee of $9+\epsilon$  based on local search \cite{kanungo}. Finally, the $k$-supplier problem is ``easier'' than both $k$-median and $k$-means in the sense that a simple $3$-approximation algorithm \cite{hochbaum} is known, as is a $2$-approximation for $k$-center problem: we cannot do better than these approximation ratios unless $\text{P} = \text{NP}$ \cite{hochbaum}. 

While optimal approximation algorithms for the center-type problems are well-known, one can easily demonstrate instances where such algorithms return a worst-possible solution: (i) all clusters have the same worst-possible radius ($2T$ for $k$-center and $3T$ for $k$-supplier where $T$ is the optimal radius) and (ii) almost all data points are on the circumference of the resulting clusters. Although it is NP-hard to improve these approximation ratios, our new randomized algorithms  provide significantly better ``per-point'' guarantees. For example, we achieve a new ``per-point'' guarantee $\bE[d(j,\S)] \leq (1+2/e)T \approx 1.736T$, while respecting the usual guarantee $d(j,\S) \leq 3T$ with probability one. \emph{Thus, while maintaining good global quality with probability one, we also provide superior stochastic guarantees for each user.}

The general problem we study in this paper is to develop approximation algorithms for center-type problems where $\S$ is drawn from a probability distribution over $k$-element subsets of $\F$; we refer to these as \emph{$k$-lotteries.} We aim to construct a $k$-lottery  $\Omega$ achieving certain guarantees on the distributional properties of $d(j,\S)$. The classical $k$-center problem can be viewed as the special case where the distribution $\Omega$ is deterministic, that is, it is supported on a single point.
 Our goal is to find an \emph{approximating distribution} $\tilde \Omega$ which matches the \emph{target distribution} $\Omega$ as closely as possible for each client $j$.
 
Stochastic solutions can circumvent the approximation hardness of a number of classical center-type problems. There are a number of additional applications where stochasticity can be beneficial. We summarize three here: smoothing the integrality constraints of clustering, solving repeated problem instances, and achieving fair solutions.

\smallskip \noindent \textbf{Stochasticity as interpolation.} In practice, \emph{robustness} of the solution is often more important than achieving the absolute optimal value for the objective function.  One potential problem with the (deterministic) center measure is that it can be highly non-robust. As an extreme example, consider $k$-center with $k$ points, each at distance $1$ from each other. This clearly has value $0$ (choosing $\S = \C$). However, if a single new point at distance $1$ to all other points is added, then the solution jumps to $1$. Stochasticity alleviates this discontinuity: by choosing $k$ facilities uniformly at random among the full set of $k+1$,  we can ensure that $\bE[d(j,\S)] = \frac{1}{k+1}$ for every point $j$, a much smoother transition. 

\smallskip \noindent \textbf{Repeated clustering problems.}
Consider clustering problems where the choice of $\mathcal{S}$ can be changed periodically: e.g., $\mathcal{S}$ could be the set of $k$ locations in the cloud chosen by a service-provider. This set $\mathcal{S}$ can be shuffled periodically in a manner transparent to end-users. For any user $j \in \C$, the statistic $d(j,\S)$ represents the latency of the service $j$ receives (from its closest service-point in $\mathcal{S}$). If we aim for a fair or minmax service allocation, then our $k$-center stochastic approximation results ensure that, with high probability, \emph{every client} $j$ gets long-term average service of at most around $1.736 T$. The average here is taken over the periodic re-provisioning of $\mathcal{S}$. (Furthermore, we have the risk-avoidance guarantee that in no individual provisioning of $\mathcal{S}$ will any client have service greater than $3T$.) 

\smallskip \noindent \textbf{Fairness in clustering.} The classical clustering problems combine the needs of many different points (elements of $\C$) into one metric. However, clustering (and indeed many other ML problems) are increasingly driven by inputs from parties with diverse interests. Fairness in these contexts has taken on greater importance in the current environment where decisions are increasingly made by algorithms and machine learning. Some examples  of recent concerns include 
the accusations of, and fixes for, possible racial bias in Airbnb rentals 
\cite{wpost:Airbnb} and 
the finding that setting the gender to ``female" in Google's \emph{Ad Settings} resulted in
getting fewer ads for high-paying jobs \cite{ad-privacy}. 
Starting with older work such as \cite{doi:10.1056/NEJM199902253400806}, there have been highly-publicized
works on bias in allocating scarce resources -- e.g., racial discrimination in hiring applicants who have very similar 
resum\'{e}s \cite{bertrand-mullainathan}. 
Additional work discusses the possibility of bias in electronic marketplaces, whether human-mediated or not \cite{RAND:RAND12115,wpost:Airbnb}.

A fair allocation should provide good service guarantees \emph{to each user individually}. In data clustering settings where a user corresponds to a datapoint, this means that every point $j \in \C$ should be guaranteed a good value of $d(j, \S)$. This is essentially the goal of $k$-center type problems, but the stochastic setting broadens the meaning of good per-user service.

Consider the following scenarios. Each user, either explicitly or implicitly, submits their data (corresponding to a point in $\C$) to an aggregator such as an e-commerce site. A small number $k$ of users are then chosen as ``influencer'' nodes; for instance, the aggregator may give them a free product sample to influence the whole population in aggregate, as in \cite{DBLP:journals/toc/KempeKT15}, or the aggregator may use them as a sparse ``sketch'', so that each user gets relevant recommendations from a influencer which is similar to them. Each point $j$ would like to be in a cluster that is ``high quality'' \emph{from its perspective}, with $d(j, \S)$ being a good proxy for such quality. Indeed, there is increasing emphasis on the fact that organizations monetize their user data, and that users need to be compensated for this (see, e.g., \cite{lanier-book,ibarra-etal:data}). This is a transition from viewing data as capital to viewing \emph{data as labor}. A concrete way for users (i.e., the data points $j \in \C$) to be compensated in our context is for each user to get a guarantee on their solution quality: i.e., bounds on $d(j, \S)$.

\subsection{Our contributions and overview}
In Section~\ref{chance-sec}, we encounter the first clustering problem which we refer to as \emph{chance $k$-coverage}: namely, every client $j$ has a distance demand $r_j$ and probability demand $p_j$, and we wish to find a distribution satisfying $\Pr[d(j,\S) \leq r_j] \geq p_j$. We show how to obtain an approximation algorithm to find an approximating distribution $\tilde \Omega$ with\footnote{Notation such as ``$S \sim \tilde \Omega$" indicates that the random set $S$ is drawn from the distribution $\tilde \Omega$.} 
$$
\Pr_{\S \sim \tilde \Omega}[d(j, \S) \leq 9 r_j] \geq p_j.
$$
In a number of special cases, such as when all the values of $p_j$ or $r_j$ are the same, the distance factor $9$ can be improved to $3$, \emph{which is optimal}; it is an interesting question to determine whether this factor  can also be improved in the general case.

In Section~\ref{sec2}, we consider a special case of chance $k$-coverage, in which $p_j = 1$ for all clients $j$. This is equivalent to the classical (deterministic) $k$-supplier problem. Allowing the approximating distribution $\tilde \Omega$ to be stochastic yields significantly better distance guarantees than are possible for $k$-supplier or $k$-center. For instance, we find an approximating distribution $\tilde \Omega$ with 
$$
\forall j \in \mathcal C \qquad \bE_{\S \sim \tilde \Omega}[d(j, \S)] \leq 1.592 T ~\mbox{and}~ \Pr[d(j, \S) \leq 3T] = 1
$$
where $T$ is the optimal solution to the (deterministic) $k$-center problem. By contrast, deterministic polynomial-time algorithms cannot guarantee $d(j, \S) < 2 T$ for all $j$, unless $\text{P} = \text{NP}$ \cite{hochbaum}. 

In Section~\ref{lb-sec}, we show a variety of lower bounds on the approximation factors achievable by efficient algorithms (assuming $\text{P} \neq \text{NP}$). For instance, we show that our approximation algorithm for  chance $k$-coverage  with equal $p_j$ or $r_j$ has the optimal distance approximation factor 3, that our approximation algorithm for $k$-supplier has optimal approximation factor $1 + 2/e$, and that the approximation factor $1.592$ for $k$-center cannot be improved below $1 + 1/e$.

In Section~\ref{approx-alg}, we consider a different type of stochastic approximation problem based on expected distances: namely, every client has a demand $t_j$, and we seek a $k$-lottery $\Omega$ with $\bE[d(j,\S)] \leq t_j$. We show that we can leverage \emph{any given} $\alpha$-approximation algorithm for $k$-median to produce a $k$-lottery $\tilde \Omega$ with 
$\bE[d(j,\S)] \leq \alpha t_j$. (Recall that the current-best $\alpha$ here is $2.675$ as shown in \cite{byrka}.) 

In Section~\ref{determinism-sec}, we consider the converse problem to Section~\ref{sec2}: if we are given a $k$-lottery $\Omega$ with $\bE[d(j,\S)] \leq t_j$, can we produce a single deterministic set $\S$ so that $d(j, \S) \approx t_j$ and $|\S| \approx k$? We refer to this as a \emph{determinization} of $\Omega$. We show a variety of determinization algorithms. For instance, we are able to find a set $\S$ with $|\S| \leq 3 k$ and $d(j, \S) \leq 3 t_j$. We also show a number of nearly-matching lower bounds.

\subsection{Related work}
With algorithms increasingly running our world, there has been substantial recent interest on incorporating fairness systematically into algorithms and machine learning. One important notion is \emph{disparate impact}: in addition to requiring that \emph{protected attributes} such as gender or race not be used (explicitly) in decisions, this asks that decisions not be disproportionately different for diverse protected classes \cite{DBLP:conf/kdd/FeldmanFMSV15}. This is developed further in the context of clustering in the work of \cite{DBLP:conf/nips/Chierichetti0LV17}. Such notions of \emph{group fairness} are considered along with \emph{individual fairness} -- treating similar individuals similarly -- in \cite{DBLP:conf/icml/ZemelWSPD13}. See \cite{DBLP:conf/innovations/DworkHPRZ12} for earlier work that developed foundations and connections for several such notions of fairness. 

In the context of the location and sizing of services, there have been several studies indicating that proactive on-site provision of healthcare improves health outcomes significantly: e.g., mobile mammography for older women \cite{reuben2002randomized}, mobile healthcare for reproductive health in immigrant women \cite{guruge2010immigrant}, and the use of a community mobile-health van for increased access to prenatal care \cite{edgerley2007use}. Studies also indicate the impact of distance to the closest facility on health outcomes: see, e.g., \cite{mccarthy2007veterans,mooney2000travel,schmitt2003influence}. 
Such works naturally suggest tradeoffs between resource allocation (provision of such services, including sizing -- e.g., the number $k$ of centers) and expected health outcomes.

While much analysis for facility-location problems has focused on the static case, other works have examined a similar lottery model for center-type problems. In~\cite{outliers, srdr}, Harris et.\ 
al.\ analyzed models similar to chance $k$-coverage and minimization of $\bE[d(j,\S)]$, but applied to knapsack center and matroid center problems; they also considered robust versions (in which a small subset of clients may be denied service). While the overall model was similar to the ones we explore here, the techniques are somewhat different. Furthermore, these works focus only on the case where the target distribution is itself deterministic.

Similar stochastic approximation guarantees have appeared in the context of approximation algorithms for static problems, particularly $k$-median problems. In \cite{charikar-li}, Charikar \& Li discussed a randomized procedure for converting a linear-programming relaxation in which a client has \emph{fractional} distance $t_j$, into a distribution $\Omega$ satisfying $\bE_{\S \sim \Omega}[d(j, \S)] \leq 3.25 t_j$. This property can be used, among other things, to achieve a $3.25$-approximation for $k$-median. However, many other randomized rounding algorithms for $k$-median only seek to preserve the \emph{aggregate} value $\sum_j \bE[d(j,\S)]$, without our type of per-point guarantee.

We also contrast our approach with a different stochastic $k$-center problem considered in works such as \cite{huang, alipour}. These consider a model with a fixed, deterministic set $\S$ of open facilities, while the client set is determined stochastically; this model is almost precisely opposite to ours.

\subsection{Publicly verifying the distributions}
Our approximation algorithms will have the following structure: given some target distribution $\Omega$, we construct a randomized procedure $\mathcal A$ which returns some random set $\S$ with good probabilistic guarantees matching $\Omega$. Thus the algorithm $\mathcal A$ is \emph{itself} the approximating distribution $\tilde \Omega$. 

In a number of cases, we can convert the randomized algorithm $\mathcal A$ into a distribution $\tilde \Omega$ which has a sparse support (set of points to which it assigns nonzero probability), and which can be enumerated directly. This may cause a small loss in approximation ratio. The distribution $\tilde \Omega$ can be publicly verified, and the users can then draw from $\tilde \Omega$ as desired.

Recall that one of our main motivations is fairness in clustering; the ability for the users to verify that they are being treated fairly in a stochastic sense (although not necessarily in any one particular run of the algorithm) is particularly important.

\subsection{Notation}
We define $\binom{\F}{k}$ to be the collection of $k$-element subsets of $\F$. We assume throughout that $\F$ can be made arbitrarily large by duplicating its elements; thus, whenever we have an expression like $\binom{\F}{k}$, we assume without loss of generality that $|\F| \geq k$. 

We will let $[t]$ denote the set $\{1, 2, \ldots, t\}$.  For any vector $a = (a_1, \dots, a_t)$ and a subset $X \subseteq [t]$, we write $a(X)$ as shorthand for $\sum_{i \in X} a_i$.

We use the Iverson notation throughout, so that for any Boolean predicate $\mathcal P$ we let $[[\mathcal P]]$ be equal to one if $\mathcal P$ is true and zero otherwise. 

For a real number $q \in [0,1]$, we define $\overline q = 1 - q$.

Given any $j \in \C$ and any real number $r \geq 0$, we define the ball $B(j, r) =  \{ i \in \F \mid d(i,j) \leq r \}$.
We let $\theta(j)$ be the distance from $j$ to the nearest facility, and $V_j$ be the facility closest to $j$, i.e. $d(j, V_j) = d(j, \F) = \theta(j)$. Note that in the SCC setting we have $V_j = j$ and $\theta(j) = 0$.

For a solution set $\S$, we say that $j \in \C$ is \emph{matched} to $i \in \S$ if $i$ is the closest facility of $\S$ to $j$; if there are multiple closest facilities, we take $i$ to be one with least index.

\subsection{Some useful subroutines}
We will use two basic subroutines repeatedly: \emph{dependent rounding} and \emph{greedy clustering}.

In dependent rounding, we aim to preserve certain marginal distributions and negative correlation properties while satisfying some constraints with probability one. Our algorithms use a dependent-rounding algorithm from \cite{srin:level-sets}, which we summarize as follows:
\begin{proposition}
\label{prop:dep-round}
There exists a randomized polynomial-time algorithm $\textsc{DepRound}(y)$ which takes as input a vector $y\in [0,1]^n$, and outputs a random set $Y \subseteq [n]$  with the following properties:
\begin{enumerate}
\item[(P1)] $\Pr[i \in Y]=y_i$, for all $i\in [n]$,
\item[(P2)] $\lfloor \sum_{i=1}^n y_i \rfloor \leq |Y| \leq \lceil \sum_{i=1}^n y_i \rceil$ with probability one, 
\label{depround:cardinality}
\item[(P3)] $\Pr[Y \cap S = \emptyset] \leq \prod_{i \in S} (1 - y_i)$ for all $S \subseteq [n]$.
\label{depround:negcorr}
\end{enumerate}
\end{proposition} 

We adopt the following additional convention: suppose $(y_1, \dots, y_n) \in [0,1]^n$ and $S \subseteq [n]$; we then define $\textsc{DepRound}(y,S) \subseteq S$ to be $\textsc{DepRound}(x)$, for the vector $x$ defined by $x_i = y_i [[ i \in S]]$.

The greedy clustering procedure takes an input a set of weights $w_j$ and sets $F_j \subseteq \F$ for every client $j \in \C$, and executes the following procedure:
\begin{algorithm}[H]
\caption{\sc GreedyCluster$(F, w)$}
\begin{algorithmic}[1]
\STATE Sort $\C$ as $\C = \{j_1, j_2, \dots, j_{\ell} \}$ where $w_{j_1} \leq w_{j_2} \leq \dots \leq w_{j_{\ell}}$.
\STATE Initialize $C' = \emptyset$
\FOR{$t = 1, \dots, \ell$}
\STATE\textbf{if} $F_{j_t} \cap F_{j'} = \emptyset$ for all $j' \in C'$ \textbf{then} update $C' \leftarrow C' \cup \{j_t \}$
\ENDFOR
\STATE Return $C'$
\end{algorithmic}
\end{algorithm}

\begin{observation}
\label{greedy-prop}
If $C' = \textsc{GreedyCluster}(F, w)$ then for any $j \in \C$ there is $z \in C'$ with $w_z \leq w_j$ and $F_z \cap F_j \neq \emptyset$. 
\end{observation}

\section{The chance $k$-coverage problem}
\label{chance-sec}
In this section, we consider a scenario we refer to as the \emph{chance $k$-coverage problem}: every point $j \in \C$ has demand parameters $p_j, r_j$, and we wish to find a $k$-lottery $\Omega$ such that
\begin{equation}
\label{stt1}
\Pr_{\S \sim \Omega}[ d(j, \S) \leq r_j] \geq p_j.
\end{equation}

If a $k$-lottery satisfying (\ref{stt1}) exists, we say that parameters $p_j, r_j$ are \emph{feasible.}  We refer to the special case wherein every client $j$ has a common value $p_j = p$ and a common value $r_j = r$, as \emph{homogeneous}. Homogeneous instances naturally correspond to fair allocations, for example,  $k$-supplier is a special case of the homogeneous chance $k$-coverage problem, in which $p_j= 1$ and $r_j$ is equal to the optimal $k$-supplier radius. 

Our approximation algorithms for this problem will be based on a linear programming (LP) relaxation which we denote $\mathcal P_{\text{chance}}$. It has fractional variables $b_i$, where $i$ ranges over $\F$ ($b_i$ represents the probability of opening facility $i$), and is defined by the following constraints:
\begin{enumerate}[label=(B\arabic*)]
	\item $\sum_{i \in B(j, r_j)} b_i  \ge p_j$ for all $j \in \C$,
\label{cmkcP:masspj}
        \item $b(\F) = k$,
\label{cmkcP:cardinality}
	\item $b_i \in [0,1]$ for all $i \in \F$.

\label{cmkcP:unit}
\end{enumerate} 

\begin{proposition}
If parameters $p, r$ are feasible, then $\P_{\text{chance}}$ is nonempty. 
\end{proposition} 
\begin{proof}
Consider a distribution $\Omega$ satisfying (\ref{stt1}).  For each $i \in \F$, set $b_i = \Pr_{\S \sim \Omega}[i \in \S]$. For $j \in \C$ we have $p_j = \Pr \bigl[ \bigvee_{i \in B(j, r_j)} i \in \S \bigr] \leq \sum_{i \in B(j,r_j)} \Pr[i \in S] =  \sum_{i \in B(j,r_j)} b_i$ and thus (B1) is satisfied. We have $k = \bE[|\S|] = \sum_{i \in \F} \Pr[i \in S] = b(\F)$ and so (B2) is satisfied. (B3) is clear, so we have demonstrated a point in $\P_{\text{chance}}$.
\end{proof}

For the remainder of this section, we assume we have a vector $b \in \P_{\text{chance}}$ and focus on how to round it to an integral solution.  By a standard facility-splitting step,  we also generate, for every $j \in \C$, a set $F_j \subseteq B(j,r_j)$ with $b(F_j) = p_j$. We refer to this set $F_j$ as a \emph{cluster}. In the SCC setting, it will also be convenient to ensure that $j \in F_j$ as long as $b_j \neq 0$.

As we show in Section~\ref{lb-sec}, any approximation algorithm must either significantly give up a guarantee on the distance, or probability (or both).  Our first result is an approximation algorithm which respects the distance guarantee exactly, with constant-factor loss to the probability guarantee:
\begin{theorem}
\label{ud1}
If $p, r$ is feasible then one may efficiently construct a $k$-lottery $\Omega$ satisfying
$$
\Pr_{\S \sim \Omega} [d(j, \S) \leq r_j] \geq (1 - 1/e) p_j.
$$
\end{theorem}
\begin{proof}
Let $b \in \P_{\text{chance}}$ and set $\S = \textsc{DepRound}(b)$. This satisfies $|\S| \leq \lceil \sum_{i=1}^n b_i \rceil \leq \lceil k \rceil = k$ as desired. Each $j \in\C$ has
$$
\Pr[\S \cap F_j = \emptyset] \leq \prod_{i \in F_j}(1-b_i) \le \prod_{i\in F_j}e^{-b_i}= e^{-b(F_i)}= e^{-p_j}.
$$
and then simple analysis shows that
\[
\Pr[ d(j, \S) \leq r_j] \geq \Pr[ \S \cap F_j \neq \emptyset] \geq 1 - e^{-p_j} \geq (1-1/e) p_j \qedhere
\]
\end{proof}
As we will later show in Theorem~\ref{ath0},  this approximation constant $1 - 1/e$ is optimal.  

We next turn to preserving the probability guarantee exactly with some loss to distance guarantee. As a warm-up exercise, let us consider the special case of ``half-homogeneous'' problem instances: all the values of $p_j$ are the same, or all the values of $r_j$ are the same. A similar algorithm works both these cases: we first select a set of clusters according to some greedy order, and then open a single item from each cluster.   We  summarize this as follows:
\begin{algorithm}[H]
\begin{algorithmic}[1]
\STATE Set $C' = \textsc{GreedyCluster}(F_j, w_j)$
\STATE Set $Y = \textsc{DepRound}(p, C')$
\STATE Form solution set $\S = \{ V_j \mid j \in Y \}$.
\end{algorithmic}
\caption{Rounding algorithm for half-homogeneous chance $k$-coverage}
\label{kcov1}
\end{algorithm}

Algorithm~\ref{kcov1} opens at most $k$ facilities, as the dependent rounding step ensures that $|Y| \leq \lceil \sum_{j \in C'} p_j \rceil = \lceil \sum_{j \in C'} b(F_j) \rceil \leq \lceil \sum_{i \in \F} b_i \rceil \leq k$.  The only difference between the two cases is the choice of weighting function $w_j$ for the greedy cluster selection.
\begin{proposition}
\label{ud2}
\begin{enumerate}
\item  Suppose that $p_j$ is the same for every $j \in \C$. Then using the weighting function $w_j = r_j$ ensures that every $j \in \C$ satisfies $\Pr[d(j,\S) \leq 3 r_j] \geq p_j.$ Furthermore, in the SCC setting, it satisfies $\Pr[d(j,\S) \leq 2 r_j] \geq p_j.$
\item Suppose $r_j$ is the same for every $j \in \C$. Then using the weighting function $w_j = 1 - p_j$ ensures that every $j \in \C$ satisfies $\Pr[d(j,\S) \leq 3 r_j] \geq p_j.$ Furthermore, in the SCC setting, it satisfies $\Pr[d(j,\S) \leq 2 r_j] \geq p_j.$
\end{enumerate}
\end{proposition}
\begin{proof}
Let $j \in \C$. By Observation~\ref{greedy-prop} there is $z \in C'$ with $w_z \leq w_j$ and $F_j \cap F_z \neq \emptyset$. In either of the two cases, this implies that $p_z \geq p_j$ and $r_z \leq r_j$.

Letting $i \in F_j \cap F_z$ gives $d(j,z) \leq d(j,i) + d(z,i) \leq r_j + r_z \leq 2 r_j$. This $z \in C'$ survives to $Y$ with probability $p_z \geq p_j$, and in that case we have $d(z,\S) = \theta(z)$. In the SCC setting, this means that $d(z,\S) = 0$; in the general setting, we have $\theta(z) \leq r_z \leq r_j$. 
\end{proof}

\subsection{Approximating the general case}
We next  show how to satisfy the probability constraint exactly for the general case of chance $k$-coverage, with a constant-factor loss in the distance guarantee. Namely, we will find a probability distribution with
$$
\Pr_{\S \sim \Omega} [d(j, \S) \leq 9 r_j] \geq p_j.
$$

The algorithm is based on iterated rounding, in which the entries of $b$ go through an unbiased random walk until $b$ becomes integral (and, thus corresponds to a solution set $\S$). Because the walk is unbiased, the probability of serving a client at the end is equal to the fractional probability of serving a client, which will be at least $p_j$. In order for this process to make progress, the number of active variables must be greater than the number of active constraints. We ensure this by periodically identifying and discarding clients which will be automatically served by serving other clients.  This is similar to a method of \cite{krishnaswamy}, which also uses iterative rounding for (deterministic) approximations to $k$-median with outliers and $k$-means with outliers.

The sets $F_j$ will remain fixed during this procedure.   We will maintain a vector $b \in [0,1]^{\F}$ and maintain two sets $\ct$ and $\cs$ with the following properties:
\begin{enumerate}
\item[(C1)] $\ct \cap \cs = \emptyset$.
\item[(C2)] For all $j, j' \in \ct$, we have $F_j \cap F_{j'} = \emptyset$
\item[(C3)] Every $j \in \ct$ has $b(F_j) = 1$,
\item[(C4)] Every $j \in \cs$ has $b(F_j) \leq 1$.
\item[(C5)] We have $b(\bigcup_{j \in \ct \cup \cs} F_j) \leq k$
\end{enumerate}

Given our initial solution $b$ for $\P_{\text{chance}}$, setting $\ct = \emptyset, \cs = \C$ will satisfy criteria (C1)--(C5); note that (C4) holds as $b(F_j) = p_j \leq 1$ for all $j \in \C$.

\begin{proposition}
\label{ggcor}
Given any vector $b \in [0, 1 ]^{\F}$ satisfying constraints (C1)---(C5) with $\cs \neq \emptyset$, it is possible to generate a random vector $b' \in [0,1]^{\F}$ such that  $\bE[b'] = b$, and with probability one $b'$ satisfies constraints (C1) --- (C5) as well as having some $j \in \cs$ with $b'(F_j) \in \{0, 1 \}.$
\end{proposition}
\begin{proof}
We will show that any basic solution $b \in [0, 1 ]^{\F}$ to the constraints (C1)---(C5) with $\cs \neq \emptyset$ must satisfy the condition that $b(F_j) \in \{0, 1 \}$ for some $j \in \cs$.  To obtain the stated result, we simply modify $b$ until it becomes basic by performing an unbiased walk in the nullspace of the tight constraints. 

So consider a basic solution $b$.  Define $A = \bigcup_{j \in \ct} F_j$ and $B = \bigcup_{j \in \cs} F_j$. We assume that $b(F_j)  \in (0,1)$ for all $j \in \cs$, as otherwise we are done.

First, suppose that $b(A \cap B) > 0$. So there must be some pair $j \in \cs, j' \in \ct$ with $i \in F_j \cap F_{j'}$ such that $b_i > 0$. Since $b(F_{j'}) = 1$, there must be some other $i' \in F_{j'}$ with $b_{i'} > 0$.  Consider modifying $b$ by incrementing $b_i$ by $\pm \epsilon$ and decrementing $b_{i'}$ by $\pm \epsilon$, for some sufficiently small $\epsilon$. Constraint (C5) is preserved. Since $F_{j'} \cap F_{j''} = \emptyset$ for all $j'' \in \ct$, constraint (C3) is preserved. Since the (C4) constraints are slack, then for $\epsilon$ sufficiently small they are preserved as well. This contradicts that $b$ is basic.

Next, suppose that $b(A \cap B) = 0$ and $b(A \cup B) < k$ strictly. Let $j \in \cs$ and $i \in F_j$ with $b_i > 0$; if we change $b_i$ by $\pm \epsilon$ for sufficiently small $\epsilon$, this preserves (C4) and (C5); furthermore, since $i \notin A$, it preserves (C3) as well. So again $b$ cannot be basic.

Finally, suppose that $b(A \cap B) = 0$ and $b(A \cup B) = k$. So $b(B) = k - |A|$ and $b(B) > 0$ as $C_{\text{slack}} \neq \emptyset$. Therefore, there must be at least two elements $i, i' \in B$ such that $b_i > 0, b_{i'} > 0$. If we increment $b_i$ by $\pm\epsilon $ while decrementing $b_{i'}$ by $\pm \epsilon$, this again preserves all the constraints for $\epsilon$ sufficiently small, contradicting that $b$ is basic.
\end{proof}

We can now describe our iterative rounding algorithm, Algorithm~\ref{algo:it3}.
\begin{algorithm}[H]
\caption{Iterative rounding algorithm for chance $k$-coverage}
\begin{algorithmic}[1]
\STATE Find a fractional solution $b$ to $\P_{\text{chance}}$ and form the corresponding sets $F_j$.
\STATE Initialize $\ct = \emptyset, \cs = \C$
\WHILE{ $\cs \neq \emptyset$ } 
\STATE Draw a fractional solution $b'$ with $\bE[b'] = b$ according to Proposition~\ref{ggcor}.
\STATE Select some $v \in \cs$ with $b'(F_v) \in \{0, 1 \}$.
\STATE Update $\cs \leftarrow \cs - \{ v \}$
\IF{$b'(F_v) = 1$}
\STATE Update $\ct \leftarrow \ct \cup \{ v \}$
\IF{there is any $z \in \ct \cup \cs$ such that $r_{z} \geq r_v/2$ and $F_{z} \cap F_v \neq \emptyset$}
\STATE{Update $\ct \leftarrow \ct - \{ z \}, \cs \leftarrow \cs - \{ z \}$}
\ENDIF
\ENDIF
\STATE Update $b \leftarrow b'$.
\ENDWHILE
\STATE For each $j \in \ct$, open an arbitrary center in $F_j$.
\end{algorithmic} 
\label{algo:it3}
\end{algorithm}

To analyze this algorithm, define $\ct^t, \cs^t, b^t$ to be the values of the relevant variables at iteration $t$. Since every step removes at least one point from $\cs$, this process must terminate in $T \leq n$ iterations. We will write $b^{t+1}$ to refer to the random value $b'$ chosen at step (4) of iteration $t$, and $v^t$  denote the choice of $v \in \cs$ used step in step (5) of iteration $t$.

\begin{proposition}
\label{bprop}
The vector $b^t$ satisfies constraints (C1) --- (C5) for all times $t = 1, \dots, T$.
\end{proposition}
\begin{proof}
The vector $b^0$ does so since $b$ satisfies $\P_{\text{chance}}$. Proposition~\ref{ggcor} ensures that step (4) does not affect this.  Also, removing points from $\ct$ or $\cs$ at step (6) or (1) will not violate these constraints.

Let us check that adding $v^t$ to $\ct$ will not violate the constraints. This step only occurs if $b^{t+1}(v^t) = 1$, and so (C3) is preserved. Since we only move $v^t$ from $\cs$ to $\ct$, constraints (C1) and (C5) are preserved. Finally, to show that (C2) is preserved, suppose that $F_{v^t} \cap F_{v^s} \neq \emptyset$ for some other $v^s$ which was added to $\ct$ at time $s < t$. If $r_{v^t} \geq r_{v^s}$, then step (10) would have removed $v^t$ from $\cs$, making it impossible to enter $\ct^t$. Thus, $r_{v^t} \leq r_{v^s}$; this means that when we add $v^t$ to $\ct^t$, we also remove $v^s$ from $\ct^t$. 
\end{proof}

\begin{corollary}
Algorithm~\ref{algo:it3} opens at most $k$ facilities.
\end{corollary}
\begin{proof}
At the final step (12), the number of open facilities is equal to $|\ct|$. By Proposition~\ref{bprop}, the vector $b^T$ satisfies constraints (C1) --- (C5). So $b(F_j) = 1$ for $j \in \ct$ and $F_j \cap F_{j'} = \emptyset$ for $j, j \in \ct$; thus $|\ct| = \sum_{j \in \ct} b(F_j) \leq k$.
\end{proof}

\begin{proposition}
\label{hh1}
If $j \in \ct^t$ for any time $t$, then  $d(j,\S) \leq 3 r_j$.
\end{proposition}
\begin{proof}
Let $t$ be maximal such that $j \in \ct^t$. We show the desired claim by induction on $t$. When $t = T$, then this certainly holds as step (12) will open some facility in $F_j$ and thus $d(j,\S) \leq r_j$.

Suppose that $j$ was added into $\ct^s$, but was later removed from $\ct^{t+1}$ due to adding $z = v^t$. Thus there is some $i \in F_z \cap F_j$. When we added $j$ in time $s$, we would have removed $z$ from $\ct^s$ if $r_z \geq r_j/2$. Since this did not occur, it must hold that $r_z < r_j/2$.

Since $z$ is present in $\ct^{t+1}$, the induction hypothesis implies that $d(z,\S) \leq 3 r_z$ and so
\[
d(j,\S) \leq d(j, i) + d(i, z) + d(z, \S) \leq r_j + r_z + 3 r_z \leq r_j  + (r_j/2) (1 + 3) = 3 r_j. \qedhere
\]
\end{proof}

\begin{theorem}
Every $j \in \C$ has $\Pr[ d(j,\S) \leq 9 r_j] \geq p_j$.
\end{theorem}
\begin{proof}
We will prove by induction on $t$ the following claim: suppose we condition on the full state of Algorithm~\ref{algo:it3} up to time $t$, and that $j \in \ct^t \cup \cs^t$. Then
\begin{equation}
\label{y1}
\Pr[ d(j,\S) \leq 9 r_j] \geq b^t(F_j).
\end{equation}

At $t = T$, this is clear; since $\cs^T = \emptyset$, we must have $j \in \ct^T$, and so $d(j,\S) \leq r_j$ with probability one.  For the induction step at time $t$, note that as $\bE[b^{t+1}(F_j)] = b(F_j)$, in order to prove (\ref{y1}) it suffices to show that if we also condition on the value of $b^{t+1}$, it holds that
\begin{equation}
\label{y2}
\Pr[ d(j,\S) \leq 9 r_j \mid b^{t+1}] \geq b^{t+1}(F_j).
\end{equation}

If $j$ remains in $\ct^{t+1} \cup \cs^{t+1}$, then we immediately apply the induction hypothesis at time $t+1$. So the only non-trivial thing to check is that (\ref{y2}) will hold even if $j$ is removed from $\ct^{t+1} \cup \cs^{t+1}$.

If $j = v^t$ and $b^{t+1}(F_j) = 0$, then (\ref{y2}) holds vacuously. Otherwise, suppose that $j$ is removed from $\ct^t$ at stage (10) due to adding $z = v^t$. Thus $r_j\geq r_z/2$ and there is some $i \in F_j \cap F_z$. By Proposition~\ref{hh1}, this ensures that $d(z,\S) \leq 3 r_z$. Thus with probability one we have
$$
d(j,\S) \leq d(j,i) + d(i,z) + d(z, \S)\leq  r_j + r_z + 3 r_z \leq r_j + (2 r_j) (1 + 3) = 9 r_j.
$$

The induction is now proved. The claimed result follows since $b^0(F_j) = p_j$ and $\cs^0 = \C$.
\end{proof}

\section{Chance $k$-coverage: approximating the deterministic case}
\label{sec2}
An important special case of  $k$-coverage is where $p_j = 1$ for all $j \in \C$. Here, the target distribution $\Omega$ is just a single set $\S$ satisfying $\forall j d(j, \S) \leq r_j$. In the homogeneous case, when all the $r_j$ are equal to the same value, this is specifically the $k$-supplier problem. The usual approximation algorithm for this problem chooses a single approximating set $\S$, in which case the best guarantee available is $d(j, \S) \leq 3 r_j$. We improve the distance guarantee by constructing a $k$-lottery $\tilde \Omega$ such that $d(j, \S) \leq 3 r_j$ with probability one, and $\bE_{\S \sim \tilde \Omega}[d(j, \S)] \leq c r_j$, 
where the constant  $c$ satisfies the following bounds:
\begin{enumerate}
\item In the general case, $c = 1 + 2/e \approx 1.73576$;
\item In the homogeneous SCC setting,  $c = 1.592$.\footnote{This value was calculated using some non-rigorous numerical analysis; we describe this further in what we call ``Pseudo-Theorem''~\ref{thm:k-center-1.592}}
\end{enumerate}
We show matching lower bounds in Section~\ref{lb-sec}; the constant value $1 + 2/e$ is optimal for the general case (even for homogeneous instances), and for the third case the constant $c$ cannot be made lower than $1 + 1/e \approx 1.367$. 

We remark that this type of stochastic guarantee allows us to efficiently construct publicly-verifiable lotteries.
\begin{proposition}
Let $\epsilon > 0$. In any of the above three cases, there is an expected polynomial time procedure to convert the given distribution $\Omega$ into an explicitly-enumerated $k$-lottery $\Omega'$, with support size $O(\frac{\log n}{\epsilon^2})$, such that $\Pr_{\S \sim \Omega'} [d(j, \S) \leq 3 r_j] = 1$ and $\bE_{\S \sim \Omega'} [d(j, \S)] \leq c (1 + \epsilon) r_j$.
\end{proposition}
\begin{proof}
Take $X_1, \dots, X_t$ as independent draws from $\Omega$ for $t = \frac{6 \log n}{c \epsilon^2}$ and set $\Omega'$ to be the uniform distribution on $\{ X_1, \dots, X_t \}$. To see that $\bE_{\S \sim \Omega'} [d(j, \S)] \leq c (1+\epsilon) r_j$ holds with high probability, apply a Chernoff bound, noting that $d(j,X_1), \dots, d(j,X_t)$ are independent random variables in the range $[0, 3 r_j]$.
\end{proof}

We use a similar algorithm to Algorithm~\ref{kcov1} for this problem: we choose a covering set of clusters $C'$, and open exactly one item from each cluster. The main difference is that instead of opening the nearest item $V_j$ for each $j \in C'$, we instead open a cluster according to the solution $b$ of $\P_{\text{chance}}$.

\begin{algorithm}[H]
\begin{algorithmic}[1]
\STATE Set $C' = \textsc{GreedyCluster}(F_j, r_j)$.
\STATE Set $F_0 = \F - \bigcup_{j \in C'} F_j$; this is the set of ``unclustered'' facilities
\FOR{$j \in C'$}
	\STATE Randomly select a point $W_j \in F_j$ according to the distribution $\Pr[W_j = i] = b_i$ \\
	\textit{// This is a valid probability distribution, as $b(F_j) = 1$}
\ENDFOR
\STATE Let $\S_0 \gets \textsc{DepRound}(b, F_0)$
\STATE Return $\S = \S_0 \cup \{ W_j \mid j \in C' \}$
\end{algorithmic} 
\caption{Rounding algorithm with clusters}
\label{algo:round0}
\end{algorithm}

We will need the following technical result in order to analyze Algorithm~\ref{algo:round0}.
\begin{proposition}
\label{sec2-bound}
For any set $U \subseteq \F$, we have
$$
\Pr[\S \cap U = \emptyset] \leq \prod_{i \in U \cap F_0}(1-b_i) \prod_{j \in C'}(1-b(U \cap F_j)) \leq e^{-b(U)}.
$$ 
\end{proposition}
\begin{proof}
The set $U$ contains each $W_j$ independently with probability $b(U \cap F_j)$. The set $\S_0$ is independent of them and by (P3) we have $\Pr[U \cap \S_0 = \emptyset] \leq \prod_{i \in U \cap F_0} (1 - b_i)$. So
\begin{align*}
\Pr[ \S \cap U = \emptyset ] &\leq  \negthinspace \prod_{i \in U \cap F_0}  \negthinspace (1 - b_i) \prod_{j \in C'} (1 - b(U \cap F_j)) \leq \negthinspace \prod_{i \in U \cap F_0} e^{-b_i} \negthinspace \prod_{j \in C'}e^{-b(U \cap F_j)}  = e^{-b(U)} \qedhere
\end{align*}
\end{proof}

At this point, we can show our claimed approximation ratio for the general (non-SCC) setting:
\begin{theorem}
\label{simple-bnd-thm0}
For any $j \in \C$, the solution set $\S$ of Algorithm~\ref{algo:round0} satisfies $d(j,\S) \leq 3 r_j$ with probability one and $\bE[d(j,\S)] \leq (1 + 2/e) r_j$.
\end{theorem}
\begin{proof}
By Observation~\ref{greedy-prop}, there is some $v \in C'$ with $F_j \cap F_v \neq \emptyset$ and $r_v \leq r_j$. Letting $i \in F_j \cap F_v$, we have
$$
d(j,\S) \leq d(j,i) + d(i,v) + d(v,\S) \leq r_j + r_v + r_v \leq 3 r_j.
$$
with probability one.

If $\S \cap F_j \neq \emptyset$, then $d(j, \S) \leq r_j$. Thus, a necessary condition for $d(j,\S) > r_j$ is that $\S \cap F_j = \emptyset$. Applying Proposition~\ref{sec2-bound} with $U = F_j$ gives
$$
\Pr[ d(j,\S) > r_j ] \leq \Pr[\S \cap F_j = \emptyset] \leq e^{-b(F_j)}= e^{-1}
$$
and so $\bE[d(j,\S)] \leq r_j  + 2 r_j \Pr[d(j,\S) > r_j] \leq (1 + 2/e) r_j$.
\end{proof}

\subsection{The homogeneous SCC setting}
\label{sec:k-center-partial-clusters}
From the point of view of the target distribution $\Omega$, this setting is equivalent to the classical $k$-center problem. We may guess the optimal radius, and so we do not need to assume that the common value of $r_j$ is ``given'' to us by some external process. By rescaling, we assume without loss of generality here that $r_j = 1$ for all $j$.

To motivate the algorithm for the SCC setting (where $\C = \F$), note that if some client $j \in \C$ has some facility $i$ opened in a nearby cluster $F_{v}$, then this guarantees that $d(j, \S) \leq d(j, v) + d(v, i) \leq 3 r_j = 3$. This is what we used to analyze the non-SCC setting. But, if instead of opening facility $i$, we opened $v$ itself, then this would ensure that $d(j, \S) \leq 2$. Thus, opening the centers of a cluster can lead to better distance guarantees compared to opening any other facility. We emphasize that this is only possible in the SCC setting, as in general we do not know that $v \in \F$. 

As a warm-up exercise, we use the following Algorithm~\ref{algo:round1}, which takes a parameter $q \in [0,1]$ which we will discuss how to set shortly. We recall that we have assumed in this case that $j \in F_j$ for every $j \in \C$.  (Also, note our notational convention that $\overline q = 1 - q$; this will be used extensively in this section to simplify the formulas.)

\begin{algorithm}[H]
\begin{algorithmic}[1]

\STATE Set $C' = \textsc{GreedyCluster}(F_j)$.
\STATE Set $F_0 = \F - \bigcup_{j \in C'} F_j$; this is the set of ``unclustered'' facilities
\FOR{$j \in C'$}

	\STATE Randomly select $W_j \in F_j$ according to the distribution $\Pr[ W_j = i ] = \overline q b_i  + q [[i =j ]]$
\ENDFOR
\STATE Let $\S_0 = \textsc{DepRound}(b, F_0)$
\STATE Return $\S = \S_0 \cup \{ W_j \mid j \in C' \}$
\end{algorithmic} 
\caption{Warm-up rounding algorithm for $k$-center}
\label{algo:round1}
\end{algorithm}
This is the same as Algorithm~\ref{algo:round0}, except that some of the values of $b_i$ for $i \in F_j$ have been shifted to the cluster center $j$.  In fact, we can think of Algorithm~\ref{algo:round1} as a two-part process: we first modify the fractional vector $b$  to obtain a new fractional vector $b'$ defined by 
$$
b'_i = \begin{cases}
\overline q b_i + q & \text{if $i \in C'$} \\
\overline q b_i & \text{if $i \in F_{j} - \{j \}$ for $j \in C'$}  \\
b_i & \text{if $i \in F_0$} 
\end{cases}.
$$
and we then execute Algorithm~\ref{algo:round0} on the resulting vector $b'$. In particular, Proposition~\ref{sec2-bound} remains valid with respect to the modified vector $b'$.

\begin{theorem}
\label{simple-bnd-thm}
Let $j \in \C$. After running Algorithm~\ref{algo:round1} with $q = 0.464587$ we have $d(j,\S) \leq 3$ with probability one and $\bE[d(j,\S)] \leq 1.60793$.
\end{theorem}
\begin{proof}
Let $D = \{v \in C' \mid F_j \cap F_{v} \neq \emptyset \}$; note that $D \neq \emptyset$ by Observation~\ref{greedy-prop}. For each $v \in D \cup \{0 \}$, set $a_v = b(F_j \cap F_v)$ and observe that $a_0 + \sum_{v \in D} a_v = b(F_j) = 1$.

As before, a necessary condition for $d(j,\S) > 1$ is that $F_j \cap \S = \emptyset$. So by Proposition~\ref{sec2-bound},
\begin{align*}
\Pr[ d(j,\S) > 1 ] &\leq \Pr[ F_j \cap \S = \emptyset ] \leq \prod_{i \in F_j \cap F_0} (1-b'_i) \prod_{v \in C'} (1 - b'(F_v \cap F_j)) \\
&\leq \prod_{i \in F_j \cap F_0} (1-b_i) \prod_{v \in D} (1 - \overline q b(F_v \cap F_j)) \leq e^{-b(F_j \cap F_0)} \prod_{v \in D} (1 - \overline q b(F_v \cap F_j)) \\
&= e^{-a_0} \prod_{v \in D} e^{a_v} (1 - \overline q a_v) = (1/e) \prod_{v \in D} e^{a_v} (1 - \overline q a_v).
\end{align*}
where the last equality comes from the fact that $a_0 + a(D) = 1$.

Similarly, if there is some $i \in \S \cap D$, then $d(j, \S) \leq d(v,i) \leq 2$. Thus, a necessary condition for $d(j,\S) > 2$ is that $\S \cap (D \cup F_j) = \emptyset$. Applying Proposition~\ref{sec2-bound} with $U = D \cup F_j$ gives:
\begin{align*}
\Pr[ d(j,\S) > 2] &\leq \prod_{v \in (D \cup F_j) \cap F_0} (1 - b_v) \prod_{v \in C'} ( 1 - b'( (D \cup F_j) \cap F_v) ) \\
&\leq e^{-b(F_j \cap F_0)}  \prod_{v \in D} \overline q (1 - a_v) = (1/e) \prod_{v \in D} e^{a_v} \overline q (1  - a_v)
\end{align*}

Putting these together gives:
\begin{equation}
\label{yy1}
\bE[d(j,\S)] \leq 1 + 1/e \prod_{v \in D} e^{a_v} (1 - \overline q a_v) + 1/e \prod_{v \in D} e^{a_v} \overline q (1 - a_v)
\end{equation}

Let us define $s = a(D)$ and $t = |D|$. By the AM-GM inequality we have:
\begin{align*}
\bE[d(j,\S)] \leq 1 + e^{s-1} \bigl( 1 - \overline q s/t \bigr)^t + e^{s -1} \overline q^t (1 - s/t)^t 
\end{align*}

This is a function of a single real parameter $s \in [0,1]$ and a single integer parameter $t \geq 1$. Some simple analysis, which we omit here, shows that $\bE[d(j,\S)] \leq 1.60793$.
\end{proof}

We will improve on Theorem~\ref{simple-bnd-thm} through a more complicated rounding process based on greedily-chosen partial clusters. Specifically, we select cluster centers $\pi(1), \dots, \pi(n)$, wherein $\pi(i)$ is chosen to maximize $b(F_{\pi(i)} - F_{\pi(1)} - \dots - F_{\pi(i-1)})$. By renumbering $\C$, we may assume without loss of generality that the resulting permutation $\pi$ is the identity; therefore, we assume throughout this section that $\C = \F = [n]$ and for all pairs $i < j$ we have
\begin{equation}
\label{ttr1}
b(F_i - F_1 - \dots - F_{i-1}) \geq b(F_j - F_1 - \dots - F_{i-1})
\end{equation}

For each $j \in [n]$, we let $G_j = F_{j} - F_{1} - \dots - F_{j-1}$ and we define $z_j = b(G_j)$. We say that $G_j$ is a \emph{full cluster} if $z_j = 1$ and a \emph{partial cluster} otherwise. 
We note that the values of $z$ appear in sorted order $1 = z_1 \geq z_2 \geq z_3 \geq \dots \geq z_n \geq 0$.

We use the following randomized Algorithm~\ref{algo:round2} to select the centers. Here, the quantities $\qf, \qp$ (short for \emph{full} and \emph{partial}) are drawn from a joint probability distribution which we discuss later. 

\begin{algorithm}[H]
\caption{$\textsc{Round2}\left(y, z,\bigcup_{j=1}^n G_j, \qf, \qp \right)$}
\begin{algorithmic}[1]
\STATE Draw random variables $\qf, \qp$.
\STATE $Z \gets \textsc{DepRound}(z)$
\FOR{$j \in Z$}
	\STATE Randomly select $W_j \in G_j$ according to the distribution $\Pr[W_j=i] = (  \overline{q_j} y_i + q_j [[ i =j ]] )/z_j$
where $q_j$ is defined as
$$
q_j = \begin{cases}
\qf & \text{if $z_j = 1$} \\
\qp & \text{if $z_j < 1$} 
\end{cases}
$$
\ENDFOR
\STATE Return $\S = \{W _j \mid j \in Z \}$
\end{algorithmic} 
\caption{Partial-cluster based algorithm}
\label{algo:round2}
\end{algorithm}

Before the technical analysis of Algorithm~\ref{algo:round2},  let us provide some intuition.  From the point of view of Algorithm~\ref{algo:round1}, it may be beneficial to open the center of some cluster near a given client $j \in \C$ as this will ensure $d(j,\S) \leq 2$. However, there is no benefit to opening more than one such cluster center. So, we would like a significant negative correlation between opening the centers of distinct clusters near $j$. Unfortunately, the full clusters all ``look alike,'' and so it seems impossible to enforce any significant negative correlation among them.

Partial clusters break the symmetry. There is at least one full cluster near $j$, and possibly some partial clusters as well. We will create a probability distribution with significant negative correlation between the event that partial clusters open their centers and the event that full clusters open their centers. This decreases the probability that $j$ will see multiple neighboring clusters open their centers, which in turn gives an improved value of $\bE[d(j,\S)]$.

The dependent rounding in Algorithm~\ref{algo:round2} ensures that $|Z| \leq \lceil \sum_{j=1}^n z_j \rceil = \sum_{j=1}^n b(F_{j} - F_{1} - \dots - F_{j-1}) = b(\F) \leq k$, and so  $|\S| \leq k$ as required. We also need the following technical result; the proof is essentially the same as Proposition~\ref{sec2-bound} and is omitted.
\begin{proposition}
\label{sec3-bound}
For any $U \subseteq \C$, we have
$$
\Pr[\S \cap U = \emptyset \mid \qf, \qp] \leq \prod_{j = 1}^n \bigl( 1 - \overline q_j b(U \cap G_j) - q_j z_j [[j \in U]] \bigr).
$$
\end{proposition}

Our main lemma to analyze Algorithm~\ref{algo:round2} is the following:
\begin{lemma}
\label{prop31}
Let $i \in [n]$. Define $J_{\text{f}}, J_{\text{p}} \subseteq [n]$ as
\begin{align*}
J_{\text{f}}   &= \{ j \in [n] \mid F_i \cap G_j \neq \emptyset, z_j = 1\} \\
J_{\text{p}} &= \{ j \in [n] \mid F_i \cap G_j \neq \emptyset, z_j < 1\}
\end{align*}

Let $m = |J_{\text{f}}|$ and $J_{\text{p}} = \{j_1, \dots, j_t \}$ where $j_1 \leq j_2 \leq \dots \leq j_t$.   For each $\ell = 1, \dots, t+1$ define
$$
u_{\ell} = b(F_i \cap G_{j_{\ell}}) + b(F_i \cap G_{j_{\ell+1}}) + \dots + b(F_i \cap G_{j_t})
$$

Then $1 \geq u_1 \geq u_2 \geq \dots \geq u_t \geq u_{t+1} = 0$ and $m \geq 1$. Furthermore, if we condition on a fixed value of $(\qf, \qp)$ then we have
\begin{align*}
\bE[d(i,\S)] & \leq 1 +  \left( 1 -  \overline \qf \frac{\overline{ u_1}}{m} \right)^m  \prod_{\ell=1}^t (1 - \overline \qp (u_{\ell} - u_{\ell+1}))  + \left( \overline \qf \bigl( 1 - \frac{\overline{u_1}}{m} \bigr) \right)^m \prod_{\ell=1}^t (\overline{u_{\ell}} + \overline \qp u_{\ell + 1}).
\end{align*}
\end{lemma}
\begin{proof}
  For $\ell = 1, \dots, t$ we let $a_{\ell} = b(F_i \cap G_{j_{\ell}}) = u_{\ell} - u_{\ell+1}$. For $j \in J_{\text{f}}$, we let $s_j = b(F_i \cap G_j)$. 

First, we claim that $z_{j_\ell} \geq u_{\ell}$ for $\ell = 1, \dots, t$. For, by (\ref{ttr1}), we have
\begin{align*}
z_{j_\ell} &\geq b(F_i - F_{1} - \dots - F_{j_{\ell}-1}) \geq b(F_i) - \sum_{j \in J_{\text{f}}} b(F_i \cap G_j) - \sum_{v<j_{\ell}} b(F_i \cap G_v) \\
&= b(F_i) - \sum_{j \in J_{\text{f}}} b(F_i \cap G_j) - \sum_{v<{\ell}} b(F_i \cap G_{j_v}) \qquad \text{as $F_i \cap G_v = \emptyset$ for $v \notin J_{\text{p}} \cup J_{\text{f}}$} \\
&= 1 - \sum_{j \in J_{\text{f}}} s_j - \sum_{v = 1}^{\ell-1} a_v = u_{\ell}.
\end{align*}

Next, observe that if $m = 0$, we have $u_1 = \sum_{\ell = 1}^t b(F_i \cap G_{j_{\ell}}) = \sum_{j=1}^n b(F_i \cap G_j) = b(F_i) = 1$. Applying the above bound at $\ell = 1$ gives $z_{j_1} \geq u_1 = 1$. But, by definition of $J_{\text{p}}$ we have $z_{j_1} < 1$, which is a contradiction. This shows that $m \geq 1$. 

We finish by showing the bound on $\bE[d(i, \S)]$. All the probabilities here should be interpreted as conditioned on fixed values for parameters $\qf, \qp$.

A necessary condition for $d(i,\S) > 1$ is that no point in $F_i$ is open. Applying Proposition~\ref{sec3-bound} with $U = F_i$ yields
\begin{align*}
\Pr[\S \cap F_i = \emptyset] &\leq \prod_{j \in J_{\text{f}}} (1 - \overline \qf  b(F_i \cap G_j) - \qf  [[j \in F_i]])  \prod_{j \in J_{\text{p}}} (1 - \overline \qp b(F_i \cap G_j)) - \qp z_j [[j \in F_i]] ) \\
&\leq\prod_{j \in J_{\text{f}}} (1 - \overline \qf b(F_i \cap G_j))   \prod_{j \in J_{\text{p}}} (1 - \overline \qp b(F_i \cap G_j)) =\prod_{j \in J_{\text{f}}} (1 - \overline \qf s_j) \prod_{\ell = 1}^t  (1 - \overline \qp a_{\ell}) .
\end{align*}

A necessary condition for $d(i,\S) > 2$ is that we do not open any point in $F_i$, nor do we open center of any cluster intersecting with $F_i$. Applying Proposition~\ref{sec3-bound} with $U = F_i \cup J_{\text{f}} \cup J_{\text{p}}$, and noting that $z_{j_{\ell}} \leq u_{\ell}$, we get:
\begin{align*}
\Pr[\S \cap U = \emptyset]  &\leq \prod_{j \in J_{\text{f}}} (1 - {\overline \qf}  b(U \cap G_j) - \qf) \prod_{j \in J_{\text{p}}} (1 - {\overline \qp} b(U \cap G_j)) - \qp z_j )  \\
&\leq\prod_{j \in J_{\text{f}}} (1 - \overline \qf b(F_i \cap G_j) - \qf)  \prod_{\ell = 1}^t  (1 - \overline \qp b(F_i \cap G_{j_\ell}) - \qp z_{j_\ell}) \\
&=\prod_{j \in J_{\text{f}}} \overline \qf (1 - s_j) \prod_{\ell = 1}^t  (1 - \overline \qp a_{\ell} - \qp z_{j_\ell})  \leq \prod_{j \in J_{\text{f}}} \overline \qf (1 - s_j)  \prod_{\ell = 1}^t  (1 - \overline \qp a_{\ell} - \qp u_{\ell}) 
\end{align*}

Thus,
\begin{equation}
\label{yy3}
\bE[d(i,\S)] \leq 1 +\prod_{j \in J_{\text{f}}} (1 - \overline \qf s_j) \prod_{\ell = 1}^t  (1 - \overline \qp a_{\ell})  +  \prod_{j \in J_{\text{f}}} \overline \qf (1 - s_j) \prod_{\ell = 1}^t  (1 - \overline \qp a_{\ell} - \qp u_{\ell})
\end{equation}

The sets $G_j$ partition $[n]$, so $\sum_{j \in \text{f}} s_j = 1 - \sum_{\ell=1}^t a_{\ell} = 1 - u_1$. So by the AM-GM inequality, we have
\begin{equation}
\label{yy4}
\bE[d(i,\S)] \leq 1 + \left( 1 - \overline \qf \frac{\overline{ u_1}}{m} \right)^m \prod_{\ell = 1}^t  (1 - \overline \qp a_{\ell}) 
+ \left( \overline \qf \bigl( 1 - \frac{\overline{u_1}}{m} \bigr) \right)^m \prod_{\ell = 1}^t  (1 - \overline \qp a_{\ell} - \qp u_{\ell}).
\end{equation}

The claim follows as $a_{\ell} = u_{\ell} - u_{\ell+1}$.
\end{proof}

Lemma~\ref{prop31} gives an upper bound on $\bE[d(i,\S)]$ for a fixed distribution on $\qf, \qp$ and for fixed values of the parameters $m, u_1, \dots, u_t$.  By a computer search, along with a number of numerical tricks, we can obtain an upper bound over all values for $m$ and all possible sequences $u_1, \dots, u_t$ satisfying $1 = u_1 \geq u_2 \geq \dots \geq u_t$. This gives the following result:
\begin{pseudo-theorem}
\label{thm:k-center-1.592}
Suppose that $\qf, \qp$ has the following distribution:
$$
( \qf, \qp) = \begin{cases}
(0.4525, 0) & \text{with probability $p = 0.773436$} \\
(0.0480, 0.3950) & \text{with probability $1-p$}
\end{cases}.
$$
Then for all $j \in \C$ we have $d(j,\S) \leq 3$ with probability one, and $\bE[d(j,\S)] \leq 1.592$.
\end{pseudo-theorem}

We call this a ``pseudo-theorem'' because some of the computer calculations used double-precision floating point for convenience, without carefully tracking rounding errors. In principle, this could have been computed using rigorous numerical analysis, giving a true theorem. Since there are a number of technical complications in this calculation, we defer the details to Appendix~\ref{proof1592}.

\section{Lower bounds on approximating chance $k$-coverage}
\label{lb-sec}
We next show lower bounds for the chance $k$-coverage problems of Sections~\ref{chance-sec} and \ref{sec2}.  These constructions are adapted from lower bounds for approximability of $k$-median \cite{guha1999greedy}, which in turn are based on the hardness of set cover.

Formally, a set cover instance consists of a collection of sets $\mathcal B = \{S_1, \dots, S_m \}$ over a ground set $[n]$. For any set $X \subseteq [m]$ we define $S_X = \bigcup_{i \in X} S_i$. The goal is to select a collection $X \subseteq [m]$ of minimum size such that $S_X = [n]$. The minimum value of $|X|$ thus obtained is denoted by $\text{OPT}$.

We quote a result of Moshovitz \cite{moshkovitz} on the inapproximability of set cover.
\begin{theorem}[\cite{moshkovitz}]
\label{hardness-th1}
Assuming $\text{P} \neq \text{NP}$, there is no polynomial-time algorithm which guarantees a set-cover solution $X$ with $S_X = [n]$ and $|X| \leq (1 - \epsilon) \ln n \times \text{OPT}$, where $\epsilon > 0$ is any constant.
\end{theorem}

We will need a simple corollary of Theorem~\ref{hardness-th1}, which is a (slight reformulation) of the hardness of approximating max-coverage.
\begin{corollary}
\label{hardprop1}
Assuming $\text{P} \neq \text{NP}$, there is no polynomial-time algorithm which guarantees a set-cover solution $X$ with $|X| \leq \text{OPT}$ and $\bigl | S_X \bigr | \geq c n$ for any constant $c > 1 - 1/e$.
\end{corollary}
\begin{proof}  
Suppose for contradiction that $\mathcal A$ is such an algorithm. We will repeatedly apply $\mathcal A$ to solve residual instances, obtaining a sequence of solutions $X_1, X_2, \dots, $. Specifically, for iterations $i \geq 1$, we define $U_i = [n] - \bigcup_{j < i} S_{X_j}$ and  $\mathcal B_i = \{S \cap U_i \mid S \in \mathcal B \}$, and we let $X_i$ denote the output of $\mathcal A$ on instance $\mathcal B_i$.

Each $\mathcal B_i$ has a solution of cost at most $\text{OPT}$. Thus, $|X_i| \leq \text{OPT}$ and $|U_i \cap S_{X_i} | \geq c |U_i|$.  Thus $|U_{i+1}| = |U_i - S_{X_i}| \leq (1-c) |U_i|$. So, for $s = \lceil 1 + \frac{\ln n}{\ln(\frac{1}{1-c})} \rceil$ we have $U_s = \emptyset$, and so the set  $X = X_1 \cup \dots \cup X_s$ solves the original set-cover instance $\mathcal B$, and $|X| \leq \sum_{i=1}^s |X_i| \leq (1 + \frac{\ln n}{\ln (\frac{1}{1 - c})}) \text{OPT}$. 

Since $c > 1 - 1/e$, we have $\frac{1}{\ln (\frac{1}{1 - c})} \leq (1 - \Omega(1))$, which contradicts Theorem~\ref{hardness-th1}.
\end{proof}

\begin{theorem}
\label{ath0}
Assuming $\text{P} \neq \text{NP}$, there is a family of homogeneous chance $k$-coverage instances, with a feasible demand of $p_j = r_j = 1$ for all clients $j$, such that no polynomial-time algorithm can guarantee a distribution $\Omega$ with either of the following:
\begin{enumerate}
\item $\forall j \bE_{\S \sim \Omega} [d(j,\S)] \leq c r_j$ for constant $c < 1 + 2/e$
\item $\forall j \Pr_{\S \sim \Omega} [ d(j, \S) < 3 r_j ] \geq  c p_j$ for constant $c > 1 - 1/e$ 
\end{enumerate}

In particular,  approximation constants in Theorem~\ref{ud1}, Proposition~\ref{ud2}, and Theorem~\ref{simple-bnd-thm0} cannot be improved.
\end{theorem}
\begin{proof}
Consider a set cover instance $\mathcal B = \{S_1, \dots, S_m \}$. We begin by guessing the value $\text{OPT}$ (there are at most $m$ possibilities, so this can be done in polynomial time). We define a $k$-center instance with $k = \text{OPT}$ and with disjoint client and facility sets, where $\F$ is identified with $[m]$ and $\C$ is identified with $[n]$. We define $d$ by $d(i,j) = 1$ if $j \in S_i$ and $d(i,j) = 3$ otherwise.

If $X$ is an optimal solution to $\mathcal B$ then $d(j,X) \leq 1$ for all points $j \in \C$. So there exists a (deterministic) distribution with feasible demand parameters $p_j = 1, r_j = 1$.  Also, note that $d(j, \S) \in \{1, 3 \}$ with probability one for any client $j$. Thus, if  $j$ satisfies the second property $\Pr_{\S \sim \Omega} [ d(j, \S) < 3 r_j ] \geq  c p_j$ for constant $c > 1 - 1/e$, then it satisfies $\bE[d(j, \S)] \leq 1 + 2 (1 - c p_j) = 3 - 2 c = (1 + 2/e - \Omega(1)) r_j$. So it will satisfy the first property as well. So it suffices to show that no algorithm $\mathcal A$ can satisfy the first property.

Suppose that $\mathcal A$ does satisfy the first property. The resulting solution set $\S \subseteq \mathcal F$ can be regarded as a solution $X$ to the set cover instance, where  $d(j, \S) = 1 + 2 [[ j \notin S_X ]]$. Thus
$$
\sum_{j \in [n]} d(j, \S) = |S_X| + 3 (n - |S_X|),
$$
and so $|S_X| = \frac{3 n - \sum_{j \in [n]} d(j, \S)}{2}$. As $\bE[d(j, \S)] \leq c r_j = c$ for all $j$, this implies that $\bE[ |S_X| ] \geq \frac{(3 - c) n}{2}$.

After an expected constant number of repetitions of this process we can ensure that $|S_X| \geq c' n$ for some constant $c' > \frac{3 - (1 + 2/e)}{2} = 1 - 1/e$. This contradicts Corollary~\ref{hardprop1}.
\end{proof}

A slightly more involved construction applies to the homogeneous SCC setting.
\begin{theorem}
\label{ath1}
Assuming $\text{P} \neq \text{NP}$, there is a family of homogeneous SCC chance $k$-coverage instances, with a feasible demand of $p_j = r_j = 1$ for all $j$, such that no polynomial-time algorithm can guarantee a distribution $\Omega$ with either of the following:
\begin{enumerate}
\item $\forall j \bE_{\S \sim \Omega} [d(j,\S)] \leq c r_j$ for constant $c < 1 + 1/e$
\item $\forall j \Pr_{\S \sim \Omega} [ d(j, \S) < 2 r_j ] \geq  c p_j$ for constant $c > 1 - 1/e$ 
\end{enumerate}

In particular, the approximation constants in Theorem~\ref{ud1} and Proposition~\ref{ud2} cannot be improved for SCC instances, and the approximation factor $1.592$ in Pseudo-Theorem~\ref{thm:k-center-1.592} cannot be improved below $1 + 1/e$.
\end{theorem}
\begin{proof}
Consider a set cover instance $\mathcal B = \{S_1, \dots, S_m\}$, where we have guessed the value $\text{OPT} = k$. We define a $k$-center instance as follows. For each $i \in [m]$, we create an item $v_i$ and for each $j \in [n]$ we create $t = n^2$ distinct items $w_{j,1}, \dots, w_{j,t}$. We define the distance by setting $d(v_i, w_{j,t}) = 1$ if $j \in S_i$ and $d(v_i, v_{i'}) = 1$ for all $i, i' \in [m]$, and $d(x,y) = 2$ for all other distances. This problem size is polynomial (in $m,n$), and so $\mathcal A$ runs in time $\text{poly}(m,n)$.

If $X$ is an optimal solution to the set cover instance, the corresponding set $\S = \{ v_i \mid i \in X \}$ satisfies $d(j, \S) \leq 1$ for all $j \in \C$. So the demand vector $p_j = r_j = 1$ is feasible.  Also, note that that $d(j, \S) \in \{0, 1, 2 \}$ with probability one for any $j$. Thus, if $j$ satisfies the second property $\Pr_{\S \sim \Omega} [ d(j, \S) < 2 r_j ] \geq  c p_j$ for constant $c > 1 - 1/e$, then it satisfies $\bE[d(j, \S)] \leq 1 + 1 (1 - c p_j) = (1 + 1/e - \Omega(1)) r_j$. So it will satisfy the first property as well. So it suffices to show that no algorithm $\mathcal A$ can satisfy the first property.

Suppose that algorithm $\mathcal A$ satisfies the first property. From the solution set $\S$, we construct a corresponding set-cover solution by $X = \{ i \mid v_i \in \S \}$.  For $w_{j,\ell} \notin \S$, we can observe that $d( w_{j, \ell}, \S) = 1 + [[ j \notin S_X ]]$. Therefore, we have
\begin{align*}
\sum_{j \in [n]} \sum_{\ell=1}^t d(w_{j,\ell}, \S) &\geq \sum_{j, \ell: w_{j,\ell} \notin \S} (1 + [[j \notin S_X]])  \geq \sum_{j, \ell} (1 + [[j \notin S_X]]) - 2 |\S| \geq n^2 (2 n - |S_X|) - 2 k,
\end{align*}
and so $|S_X| \geq 2 n - \frac{\sum_{j, \ell} d(w_{j,\ell}, \S)}{n^2} - 2 k/n^2.$

Taking expectations and using our upper bound on $\bE[d(j,\S)]$,  we have $\bE[ |S_X| ] \geq 2 n - c n - 2 k/n^2 \geq (2 - c) n - 2/n$. Thus, for $n$ sufficiently large, after an expected constant number of repetitions of this process we get $|S_X| \geq (2 - c - o(1))n \geq (1 - 1/e + \Omega(1)) n$. This contradicts Corollary~\ref{hardprop1}.
\end{proof}

\section{Approximation algorithm for $\bE[d(j,\S)]$}
\label{approx-alg}
In the chance $k$-coverage problem, our goal is to achieve certain fixed values of $d(j,\S)$ with a certain probability. In this section, we consider another criterion for $\Omega$; we wish to achieve certain values for the expectation $\bE_{\S \sim \Omega}[d(j, \S)]$.  We suppose we are given values $t_j$ for every $j \in C$, such that the target distribution $\Omega$ satisfies
\begin{equation}
\label{stt5}
\bE_{\S \sim \Omega} [d(j, \S)] \leq t_j.
\end{equation}

In this case, we say that the vector $t_j$ is \emph{feasible.} 
As before, if all the values of $t_j$ are equal to each other, we say that the instance is \emph{homogeneous}.  
We show how to leverage any approximation algorithm for $k$-median with approximation ratio $\alpha$, to ensure our target distribution $\tilde \Omega$ will satisfy
$$
\bE_{\S \sim \tilde \Omega} [d(j, \S)] \leq (\alpha + \epsilon) t_j.
$$

More specifically, we need an approximation algorithm for weighted $k$-median. In this setting, we have a problem instance $\I = \F, \C, d$ along with non-negative weights $w_j$ for $j \in \C$, and we wish to find  $\S \in \binom{\F}{k}$ minimizing $\sum_{j \in \C} w_j d(j, \S).$  (Nearly all approximation algorithms for ordinary $k$-median can be easily adapted to the weighted setting, for example, by  replicating clients.) If we fix an approximation algorithm $\mathcal A$ for (various classes of) weighted $k$-median, then for any problem instance $\I$ we define
$$
\alpha_{\I} = \sup_{\text{weights $w$}} \frac{ \sum_{j \in \C} w_j d(j, \mathcal A(\I, w)) }{ \min_{\S \in \binom{\F}{k}} \sum_{j \in \C} w_j d(j, \S) }.
$$

We first show how to use the $k$-median approximation algorithm to achieve a set $\S$ which ``matches'' the desired distances $t_j$:
\begin{proposition}
\label{kmedianprop1}
Given a weighted instance $\I$ and a parameter $\epsilon > 0$, there is a polynomial-time algorithm  to produce a set $\S \in \binom{\F}{k}$ satisfying:
\begin{enumerate}
\item $\sum_{j \in \C} w_j \frac{d(j, \S)}{t_j} \leq (\alpha_{\I} + O(\epsilon)) \sum_{j \in \C} w_j$,
\item Every $j \in \C$ has $d(j, \S) \leq n t_j/\epsilon$.
\end{enumerate}
\end{proposition}
\begin{proof}
We assume $\alpha_{\I} = O(1)$, as constant-factor approximation algorithms for $k$-median exist. By rescaling $w$,  we assume without loss of generality that $\sum_{j \in \C} w_j = 1$. By rescaling $\epsilon$, it suffices to show that $d(j, \S) \leq O(n t_j/\epsilon)$.

Let us define the weight vector $z_j = \frac{\epsilon/n + w_j}{t_j}$. Letting $\Omega$ be  a distribution satisfying (\ref{stt5}), we have
\begin{align*}
\bE_{\S \sim \Omega} \bigl[ \sum_{j \in \C} z_j d(j, \S) \bigr] &= \sum_{j \in \C} z_j t_j \leq \sum_{j \in \C} (\frac{\epsilon}{n t_j} + \frac{w_j}{t_j}) t_j = \epsilon | \C |/n + \sum_{j \in \C} w_j = 1 + \epsilon.
\end{align*}

In particular, there exists some $\S \in \binom{\F}{k}$ with $\sum_{j \in \C} z_j d(j, \S) \leq 1 + \epsilon$. When we apply algorithm $\mathcal A$ with weight vector $z$, we thus get a set $\S \in \binom{\F}{k}$ with $\sum_{j \in \C} z_j d(j, \S) \leq \alpha_{\I} (1 + \epsilon)$. We claim that this set $\S$ satisfies the two conditions of the theorem. First, we have
\begin{align*}
\sum_{j \in \C} \frac{w_j d(j, \S)}{t_j} &\leq \sum_{j \in \C} z_j d(j,\S) \leq \alpha_{\I} (1 + \epsilon) \leq (\alpha_{\I} + O(\epsilon)) \sum_j w_j.
\end{align*}

Next, for any given $j \in \C$, we have
\[
\frac{d(j, \S)}{t_j} \leq d(j, \S) z_j  (n/\epsilon) \leq (n/\epsilon) \sum_{w \in \C} z_w d(w,\S) \leq (n/\epsilon) \alpha_{\I} (1 + \epsilon) \leq O(n/\epsilon).  \qedhere
\]
\end{proof}
\begin{theorem}
\label{main-approx-thm}
There is an algorithm which takes as input an instance $\I$, a parameter $\epsilon > 0$ and a feasible vector $t_j$,  runs in time $\text{poly}(n,1/\epsilon)$, and returns an explicitly enumerated distribution $\tilde \Omega$ with support size $n$ and $\bE_{S \sim \tilde \Omega}[d(j, \S)] \leq (\alpha_{\I} + \epsilon) t_j$ for all $j \in \C$.
\end{theorem}
\begin{proof}
We assume without loss of generality that $\epsilon \leq 1$; by rescaling $\epsilon$ it suffices to show that $\bE[d(j,\S)] \leq (\alpha_\I + O(\epsilon)) t_j$.

We begin with the following Algorithm~\ref{algo:app0}, which uses a form of multiplicative weights update with repeated applications of Proposition~\ref{kmedianprop1}.
\begin{algorithm}[H]
\begin{algorithmic}[1]
\FOR{$\ell = 1, \dots, r = \frac{n \ln n}{\epsilon^3}$}
\STATE Let $X_{\ell} \in \binom{\F}{k}$ be the resulting of applying the algorithm of Proposition~\ref{kmedianprop1} with parameters $\epsilon, t_j$ and weight vector $w$ given by 
\vspace{-0.16in}
$$
w_j = \exp \Bigl( \epsilon^2 \sum_{s = 1}^{\ell-1} \frac{d(j,X_s)}{n t_j} \Bigr),
$$
\vspace{-0.12in}
\ENDFOR
\STATE Set $\tilde \Omega'$ to be the uniform distribution on $X_1, \dots, X_r$
\end{algorithmic} 
\caption{Approximation algorithm for $\bE[d(j,\S)]$: first phase}
\label{algo:app0}
\end{algorithm}
Let us define $\phi = \epsilon^2/n$.  For each iteration $\ell = 1, \dots, r+1$ let $u_j^{(\ell)} =  \phi d(j,X_{\ell}) / t_j$, and let $w^{(\ell)}_j = e^{\sum_{s=1}^{\ell-1} u_j^{(s)}}$ denote the weight vector. Proposition~\ref{kmedianprop1} ensures that $u^{(\ell)}_j \leq \epsilon$, and thus $e^{u_j^{(\ell)}} \leq 1 + \frac{e^\epsilon - 1}{\epsilon} u_j^{(\ell)} \leq 1 + (1+\epsilon) u_j^{(\ell)}$, as well as ensuring that  $\sum_j w_j^{(\ell)} u_j^{(\ell)} \leq \phi (\alpha_I + O(\epsilon)) \sum_{j} w_j^{(\ell)}$.

Now let $\Phi_{\ell} = \sum_{j \in \C} w^{(\ell)}_j$.  Note that $\Phi_1 =n$, and for each $\ell \geq 1$, we have
\begin{align*}
\Phi_{\ell+1} &= \sum_{j \in \C} w^{(\ell)}_j e^{u_j^{(\ell)}} \leq \sum_{j \in \C} w^{(\ell)}_j \bigl( 1 + (1 + \epsilon) u_j^{(\ell)} \bigr) \leq \sum_{j \in \C} w^{(\ell)}_j + (1 + \epsilon) \sum_{j \in \C} w^{(\ell)}_j u_j^{(\ell)} \\
&\leq \Phi_{\ell} \bigl( 1 + (1 + \epsilon) \phi (\alpha_{\I} + O(\epsilon)) \bigr) \leq \Phi_{\ell} e^{ \phi (\alpha_{\I} + O(\epsilon))}
\end{align*}

This recurrence relation implies that $\Phi_{\ell} \leq n e^{ (\ell-1) \phi (\alpha_{\I} + O(\epsilon))  }$. Since $w^{(r+1)}_j \leq \Phi^{r+1}$,  this implies
$$
 \sum_{\ell=1}^r \phi d(j,X_{\ell}) / t_j  = \ln w_j^{(r+1)}  \leq \ln \Phi_{r+1} \leq \ln n + r \phi ( \alpha_{\I} + O(\epsilon)).
$$
or equivalently,
$$
 \sum_{\ell=1}^r \frac{d(j,X_{\ell})}{r} = t_j \Bigl( \frac{\ln n}{r \phi} + (\alpha_{\I} + O(\epsilon)) \Bigr)
$$
As $r = \frac{n \ln n}{\epsilon^3} = \frac{\ln n}{\epsilon \phi}$, we thus have $\frac{\sum_{\ell=1}^r d(j,X_{\ell})}{r} \leq (\alpha_{\I} + O(\epsilon)) t_j$. Thus, the distribution $\tilde \Omega'$ satisfies
\begin{equation}
\label{bv13}
\forall j \in \mathcal C \qquad \bE_{\S \sim  \tilde \Omega'}[d(j, \S)] \leq  (\alpha_I + O(\epsilon)) t_j.
\end{equation}

Now the distribution $\tilde \Omega'$ satisfies the condition on $\bE[d(j, \S)]$, but its support is too large. We can reduce the support size to $|\C|$ by moving in the nullspace of the $| \C |$ linear  constraints (\ref{bv13}).
\end{proof}

Byrka et al. \cite{byrka} have shown a $2.675 + \epsilon$-approximation algorithm for $k$-median, which automatically gives a $2.675 + \epsilon$-approximation algorithm for $k$-lottery as well. 
Some special cases of $k$-median have more efficient approximation algorithms. For instance, Cohen-Addad, Klein \& Mathieu \cite{cohen2016local} gives a PTAS for $k$-median problems derived from a planar graph, and Ahmadian et al. \cite{svensson} gives a $2.633 + \epsilon$-approximation for Euclidan distances. These immediately give approximation algorithms for the corresponding $k$-lotteries. We also note that, by Theorem~\ref{ath0}, one cannot obtain a general approximation ratio better than $1 + 2/e$ (or $1 + 1/e$ in the SCC setting).

\section{Determinizing a $k$-lottery}
\label{determinism-sec}
Suppose that we have a set of feasible weights $t_j$ such some $k$-lottery distribution $\Omega$ satisfies $\bE_{\S \sim \Omega }[d(j, \S)] \leq t_j$; let us examine how to find a \emph{single, deterministic} set $\S$ with $d(j, \S) \approx t_j$. We refer to this as the problem of \emph{determinizing} the lottery $\Omega$. Note that this can be viewed as a converse to the problem considered in Section~\ref{sec2}.

We will see that, in order to obtain reasonable approximation ratios, we may need $|\S|$ to be significantly larger than $k$. We thus define an \emph{$(\alpha, \beta)$-determinization} to be a set $\S \in \binom{\F}{k'}$ with $k' \leq \alpha k$ and  $d(j, \S) \leq \beta t_j$ for all $j \in \C$. We emphasize that we cannot necessarily obtain $(1,1)$-determinizations, even with unbounded computational resources. The following simple example illustrates the tradeoff between parameters $\alpha$ and $\beta$:
\begin{observation}
\label{lb-obs}
Let $\alpha, \beta, k \geq 1$. If $\beta < \frac{\alpha k + 1}{(\alpha - 1) k + 1}$, there is a homogeneous SCC instance for which no $(\alpha, \beta)$-determinization exists.
\end{observation}
\begin{proof}
Let $k' = \alpha k$ and consider a problem instance with $\F = \C = \{1, \dots, k'+1 \}$, and $d(i,j) = 1$ for every distinct $i,j$. Clearly, every $\S \in \binom{\F}{k'}$ satisfies $\min_j d(j, \S) = 1$. When $\Omega$ is the uniform distribution on $\binom{\F}{k}$, we have $\bE[d(j, \S)] = 1 - \frac{k}{k'+1}$. Thus $t_j = \frac{k}{k'+1}$ is feasible and therefore $\beta \geq \frac{1}{1 - \frac{k}{k'+1}} = \frac{ \alpha k + 1}{(\alpha - 1) k + 1}$.
\end{proof}

In particular, when $\alpha = 1$ we must have $\beta \geq k+1$ and when $k \rightarrow \infty$, we must have $\beta \gtrsim \frac{\alpha}{\alpha - 1}$.

We examine three main regimes for the parameters $(\alpha, \beta)$: (1) the case where $\alpha, \beta$ are scale-free constants; (2) the case where $\beta$ is close to one, in which case $\alpha$ must be of order $\log n$; (3) the case where $\alpha = 1$, in which case $\beta$ must be order $k$.

Our determinization algorithms for the first two cases will based on the following LP denoted $\mathcal P_{\text{expectation}}$, defined in terms of fractional vectors $b_i, a_{i,j}$ where $i$ ranges over $\F$ and $j$ ranges over $\C$:
\begin{enumerate}
\item[(A1)] $\forall j \in \C, \qquad \sum_{i \in \F} a_{i,j} d(i,j) \leq t_j$,
\item[(A2)] $\forall j \in \C, \qquad \sum_{i \in \F} a_{i,j} = 1$,
\item[(A3)] $\forall i \in \F, y \in \C, \qquad 0 \leq a_{i,j} \leq b_i$,
\item[(A4)] $\forall i \in \F, \qquad 0 \leq b_i \leq 1$,
\item[(A5)] $\sum_{i \in \F} b_i \leq k$.
\end{enumerate}

\begin{theorem}
If $t_j$ is feasible, then $\mathcal P_{\text{expectation}}$ has a fractional solution.
\end{theorem}
\begin{proof}
Let $\Omega$ be a probability distribution with $\bE[d(j,\S)] \leq t_j$. For any draw $\S \sim \Omega$, define random variable $Z_{j}$ to be the facility of $\S$ matched by $j$. Now consider the fractional vector defined by
$$
b_i = \Pr_{\S \sim \Omega} [ i \in \S ], \qquad \qquad a_{i,j} = \Pr_{\S \sim \Omega} [ Z_{j} = i ]
$$

We claim that this satisfies (A1) --- (A5). For (A1), we have
$$
\bE[d(j, \S)] = \bE[d(j, Z_{j})] =  \sum_{i \in \F} d(i,j) \Pr[Z_{j} = i] = \sum_{i \in F} d(i,j) a_{i,j} \leq t_j.
$$

For (A2), note that $\sum_i \Pr[Z_j = i] = 1$.  For (A3), note that $Z_j = i$ can only occur if $i \in \S$. (A4) is clear, and (A5) holds as $|\S| = k$ with probability one. 
\end{proof}

We next describe upper and lower bounds for these three regimes.

\subsection{The case where $\alpha, \beta$ are scale-free constants.} 
For this regime (with all parameters independent of problem size $n$ and $k$), we may use the following Algorithm~\ref{algo:det1}, which is based on greedy clustering using a solution to $\mathcal P_{\text{expectation}}$.

\begin{algorithm}[H]
\begin{algorithmic}[1]
\STATE Let $a,b$ be a solution to $\mathcal P_{\text{expectation}}$.
\STATE For every $j \in \C$, select $r_j \geq 0$ to be minimal such that $\sum_{i \in B(j, r_j)} a_{i,j} \geq 1/\alpha$
\STATE By splitting facilities, form a set $F_j \subseteq B(j, r_j)$ with $b(F_j) = 1/\alpha$.
\STATE Set $C' = \textsc{GreedyCluster}(F_j, \theta(j) + r_j)$
\STATE Output solution set $\S = \{ V_j \mid j \in C' \}$.
\end{algorithmic} 
\caption{$(\alpha,\beta)$-determinization algorithm}
\label{algo:det1}
\end{algorithm}

Step (3) is well-defined, as (A3) ensures that $b(B(j, r_j)) \geq \sum_{i \in B(j,r_j)} a_{i,j} \geq 1/\alpha$. Let us analyze the resulting approximation factor $\beta$.

\begin{proposition}
\label{rybound}
Every client $j \in \C$ has $r_j \leq \frac{\alpha t_j - \theta(j)}{\alpha - 1}$.
\end{proposition}
\begin{proof}
Let  $s = \frac{\alpha t_j - \theta(j)}{\alpha - 1}$. It suffices to show that
$$
\sum_{i \in \F, d(i,j) > s} a_{i,j} \leq 1 - 1/\alpha.
$$

As $d(i,j) \geq \theta(j)$ for all $i \in \F$, we have
{\allowdisplaybreaks
\begin{align*}
\sum_{\substack{i \in \F \\ d(i,j) > s}} a_{i,j} &\leq \sum_{\substack{i \in \F \\ d(i,j) > s}} a_{i,j} \frac{d(i,j) - \theta(j)}{s - \theta(j)} \leq \sum_{i \in \F} a_{i,j} \frac{d(i,j) - \theta(j)}{s - \theta(j)} = \frac{\sum_{i \in \F} a_{i,j} d(i,j) - \theta(j) \sum_{i \in \F} a_{i,j}}{s-\theta(j)} \\
&\leq \frac{t_j - \theta(j)}{s-\theta(j)} = 1 - 1/\alpha, \qquad \text{by (A1), (A2).}  \qedhere
\end{align*}
}
\end{proof}

\begin{theorem}
\label{algo:det1-result}
Algorithm~\ref{algo:det1} gives an $(\alpha, \beta)$-determinization with the following parameter $\beta$:
\begin{enumerate}
\item In the general setting, $\beta = \max(3, \frac{2 \alpha}{\alpha - 1})$.
\item In the SCC setting, $\beta = \frac{2 \alpha}{\alpha - 1}$.
\end{enumerate}
\end{theorem}
\begin{proof}
We first claim that the resulting set $\S$ has $|\S| \leq \alpha k$. The algorithm opens at most $|C'|$ facilities. The sets $F_j$ are pairwise disjoint for $j \in C'$ and $b(F_j) = 1/\alpha$ for $j \in C'$. Thus $\sum_{j \in C'} b(F_j) = |C'|/\alpha$. On the other hand, $b(\F) = k$, and so $k \geq |C'|/\alpha$.

Next, consider some $j \in \C$; we want to show that $d(j, \S) \leq \beta t_j$. By Observation~\ref{greedy-prop}, there is $z \in C'$ with $F_j \cap F_z \neq \emptyset$ and $\theta(z) + r_z \leq \theta(j) + r_j$. Thus $d(j, \S) \leq d(z,\S) + d(z,i) + d(j,i)$ where $i \in F_j \cap F_z$. Step (5) ensures $d(z,\S) = \theta(z)$. We have $d(z,i) \leq r_z$ and $d(i,j) \leq r_j$ since $i \in F_j \subseteq B(j, r_j)$ and $i \in F_z \subseteq B(z, r_z)$. So
$$
d(j, \S) \leq \theta(z) + r_z + r_j \leq  2 r_j + \theta(j).
$$

By Proposition~\ref{rybound}, we therefore have
\begin{equation}
\label{hhg1}
d(j, \S) \leq \frac{2 \alpha t_j - 2 \theta(j)}{\alpha - 1} + \theta(j) = \frac{2 \alpha t_j}{\alpha-1} + \frac{\alpha - 3}{\alpha - 1} \theta(j)
\end{equation}

This immediately shows the claim for the SCC setting where $\theta(j)= 0$.

In the general setting, for $\alpha \leq 3$, the second coefficient in the RHS of (\ref{hhg1}) is non-positive and hence the RHS is at most $\frac{2 \alpha t_j}{\alpha-1}$ as desired. When $\alpha \geq 3$, then in order for $t$ to be feasible we must have $t_j \geq \theta(j)$; substituting this upper bound on $\theta(j)$ into (\ref{hhg1}) gives
\[
d(j, \S) \leq \frac{2 \alpha t_j}{\alpha-1} + \frac{\alpha - 3}{\alpha - 1} t_j = 3 t_j \qedhere
\]
\end{proof}

We note that these approximation ratios are, for $\alpha$ close to $1$, within a factor of $2$ compared to the lower bound of Observation~\ref{lb-obs}. As $\alpha \rightarrow \infty$,  the approximation ratio approaches to limiting values $3$ (or $2$ in the SCC setting). 

\subsection{The case of small $\beta$} We now consider what occurs when $\beta$ becomes smaller than the critical threshold values $3$ (or $2$ in the SCC setting).  We show that in this regime we must take $\alpha = \Omega(\log n)$.  Of particular interest is the case when $\beta$ approaches $1$; here, in order to get $\beta = 1 +\epsilon$ for small $\epsilon$ we show it is necessary and sufficient to take $\alpha = \Theta( \frac{\log n}{\epsilon})$.

\begin{proposition}
\label{det1prop}
For any $\epsilon < 1/2$,  there is a randomized polynomial-time algorithm to obtain a $(\frac{3 \log n}{\epsilon}, 1 + \epsilon)$ determinization.
\end{proposition}
\begin{proof}
First, let $a,b$ be a solution to $\mathcal P_{\text{expectation}}$. Define $p_i = \min(1, \frac{2 \log n}{\epsilon} b_i)$ for each $i \in \F$ and form $\S = \textsc {DepRound}(p)$. Observe then that $|\S| \leq \lceil \sum_i p_i \rceil \leq \lceil \frac{2 \log n}{\epsilon} \sum b_i \rceil \leq 1 + \frac{2 k \log n}{\epsilon} \leq \frac{3 k \log n}{\epsilon}$.

For $j \in \C$, define $A = B_{j, (1+\epsilon) t_j}$. Let us note that, by properties (A3), (A1) and (A2), we have
$$
\sum_{i \in A} b_i \geq \sum_{i \in A} a_{i,j} = 1 - \sum_{i: d(i,j) > (1+\epsilon) t_j} a_{i,j} \geq  1 - \sum_{i: d(i,j) > (1+\epsilon) t_j} a_{i,j} \frac{d(i,j)}{(1+\epsilon) t_j} \geq 1 - \frac{1}{1 + \epsilon}
$$

So by property (P3) of \textsc{DepRound}, and using the bound $\epsilon < 1/2$,  have
{\allowdisplaybreaks
\begin{align*}
\Pr[d(j, \S) > (1 + \epsilon) t_j] &= \Pr[ A \cap \S = \emptyset] \leq \prod_{i \in A} (1 - p_i) \leq \prod_{i \in A} e^{-\frac{2 \log n}{\epsilon} b_i} \leq e^{ \frac{-2 \log n}{\epsilon} (1 -  \frac{1}{1+\epsilon})} \leq n^{-4/3}
\end{align*}
}

A union bound over $j \in \C$ shows that solution set $\S$ satisfies $d(j, \S) \leq (1+\epsilon) t_j$ for all $j$ with high probability.
\end{proof}

The following shows matching lower bounds:
\begin{proposition}
\label{rr4a}
\begin{enumerate}
\item There is a universal constant $K$ with the following properties. For any $k \geq 1, \epsilon \in (0,1/3)$ there is some integer $N_{k,\epsilon}$ such that for $n > N_{k,\epsilon}$, there is a homogeneous SCC instance of size $n$ in which \emph{every} $(\alpha, 1+\epsilon)$-determinization satisfies $\alpha \geq \frac{K \log n}{\epsilon}.$
\item For each $\beta \in (1,2)$ and each $k \geq 1$, there is a constant $K'_{\beta, k}$ such that, for all $n \geq 1$, there is a homogeneous SCC instance of size $n$ in which \emph{every} $(\alpha, \beta)$-determinization satisfies $\alpha \geq K'_{\beta, k} \log n.$
\item For each $\beta \in (1,3)$ and each $k \geq 1$, there is a constant $K''_{\beta, k}$ such that, for all $n \geq 1$, there is a homogeneous  instance of size $n$ in which \emph{every} $(\alpha, \beta)$-determinization satisfies $\alpha \geq K''_{\beta, k} \log n$
\end{enumerate}
\end{proposition}
\begin{proof}
These three results are very similar, so we show the first one in detail and sketch the difference between the other two.

Consider an Erd\H{o}s-R\'{e}nyi random graph $G \sim \mathcal G(n,p)$, where $p = 3 \epsilon / k$; note that $p \in (0,1)$.  As shown by \cite{glebov} asymptotically almost surely the domination number $J$ of $G$ satisfies $J =  \Omega( \frac{k \log n}{\epsilon})$.

We construct a related instance with $\F = \C = [n]$, and where $d(i,j) = 1$ if $(i,j)$ is an edge, and $d(i,j) = 2$ otherwise.  Note that if $X$ is not a dominating set of $G$, then some vertex of $G$ has distance at least $2$ from it; equivalently, $\max_j d(j,X) \geq 2$ for every set $X$ with $|X| < J$.

Chernoff's bound shows that every vertex of $G$ has degree at least $u = 0.9 n p$ with high probability. Assuming this event has occured, we calculate $\bE[d(j, \S)]$ where $\S$ is drawn from the uniform distribution on $\binom{\F}{k}$. Note that $d(j, \S) \leq 1$ if $j$ is a neighbor of $X$ and $d(j, \S) = 2$ otherwise, so
$$
\bE[d(j, \S)] \leq 1 + \frac{ \binom{n - u}{k}}{\binom{n}{k}} \leq 1+ e^{-0.9 p k} = 1 + e^{-2.7 \epsilon}.
$$

Both the bound on the domination number and the minimum degree of $G$ hold with positive probability for $n$ sufficiently large (as a function of $k, \epsilon$). In this case, $t_j = 1 + e^{-2.7 \epsilon}$ is a feasible homogeneous demand vector. At the same time, every set $\S \in \binom{F}{J-1}$ satisfies $\min_{j \in \C} d(j, \S) \geq 2$. Thus, an $(\alpha, \beta)$-determinization cannot have $\alpha < \frac{J}{k} = \Theta(\frac{\log n}{\epsilon})$ and $\beta \leq \frac{2}{1 + e^{-2.7 \epsilon}}$. Note that $\frac{2}{1 + e^{-2.7 \epsilon}} \geq 1 + \epsilon$ for $\epsilon < 1/3$. Thus, whenever $\beta \leq 1 + \epsilon$, we have $\alpha \geq \Theta( \frac{\log n}{\epsilon} )$.

For the second result, we use the same  construction as above with $p = 1 - \tfrac{1}{2} (\lambda/2)^{1/k}$ where $\lambda = 2 - \beta$. A similar analysis shows that the vector $t_j = 1 + \lambda/2$ is feasible with high probability and $|J| \geq \Omega(k \log n)$ (where the hidden constant may depend upon $\beta, k$). Thus, unless $\alpha \geq \Omega(\log n)$, the approximation ratio achieved is $\frac{2}{1 + \lambda/2} \geq \beta$.

The third result is similar to the second one, except that we use a random bipartite graph. The left-nodes are associated with $\F$ and the right-nodes with $\C$. For $i \in \F$ and $j \in \C$, we define $d(i,j) = 1$ if $(i,j)$ is an edge and $d(i,j) = 3$ otherwise.
\end{proof}

\subsection{The case of $\alpha = 1$}
We finally consider the case $\alpha = 1$, that is, where  the constraint on the number of open facilities is respected \emph{exactly}. By Observation~\ref{lb-obs}, we must have $\beta \geq k+1$ here. The following greedy algorithm gives a $(1, k+2)$-determinization, nearly matching this lower bound. 
\begin{algorithm}[H]
\begin{algorithmic}[1]
\STATE Initialize $\S = \emptyset$
\FOR{$\ell = 1, \dots, | \F |$}
\STATE Let $\mathcal C_{\ell}$ denote the set of points $j \in \C$  with $d(j, \S) > (k+2) t_j$
\STATE If $\mathcal C_{\ell} = \emptyset$, then return $\S$.
\STATE Select the point $j_{\ell} \in \mathcal C_{\ell}$ with the smallest value of $t_{j_{\ell}}$.
\STATE Update $\S \leftarrow \S \cup \{ V_{j_{\ell}} \}$
\ENDFOR
\end{algorithmic} 
\caption{$(1, k+2)$-determinization algorithm}
\label{algo:det2}
\end{algorithm}

\begin{theorem}
If the values $t_j$ are feasible, then Algorithm~\ref{algo:det2} outputs a $(1, k+2)$-determinization in $O(|\mathcal \F| | \mathcal \C |)$ time.
\end{theorem}
\begin{proof}
For the runtime bound, we first compute $V_j$ for each $j \in \C$; this requires $O( |\mathcal F| |\mathcal C|)$ time upfront. When we update $\S$ at each iteration $\ell$, we update and maintain the quantities $d(j,\S)$  quantities by computing $d(j, V_{j_{\ell}})$ for each $j \in \C$. This takes $O(|\C|)$ time per iteration.

To show correctness, note that if this procedure terminates at iteration $\ell$, we have $\mathcal C_{\ell} = \emptyset$ and so every point $j \in \C$ has $d(j, \S) \leq (k+2) t_j$. The resulting set $\S$ at this point has cardinality $\ell - 1$. So we need to show that the algorithm terminates before reaching iteration $\ell = k+2$. 

Suppose not; let the resulting points be $j_1, \dots, j_{k+1}$ and for each $\ell = 1, \dots, k+1$ let $w_{\ell} = t_{j_\ell}$.  Because $j_{\ell}$ is selected to minimze $t_{j_{\ell}}$ we have $w_1 \leq w_2 \leq \dots \leq w_{k+1}$.

Now, let $\Omega$ be a $k$-lottery satisfying $\bE_{\S \sim \Omega}[ d(j, \S)] \leq t_j$ for every $j \in \mathcal C$, and consider the random process of drawing $\S$ from $\Omega$. Define the random variable $D_{\ell} = d(j_{\ell}, \S)$ for $\ell =1, \dots, k+1$. For any such $\S$, by the pigeonhole principle there must exist some pair $j_{\ell}, j_{\ell'}$ with $1 \leq \ell < \ell' \leq k+1$ which are both matched to a common facility $i \in \S$, that is
$$
D_{\ell} = d(j_{\ell}, \S) = d(j_{\ell}, i), D_{\ell'} = d(j_{\ell'}, \S) = d(j_{\ell'}, i).
$$

By the triangle inequality, 
$$
d(j_{\ell'}, V_{j_{\ell}}) \leq d(j_{\ell'}, i) + d(i, j_{\ell}) + d(j_{\ell}, V_{j_{\ell}}) = D_{\ell'} + D_{\ell} + \theta(j_{\ell})
$$

On the other hand, $j_{\ell'} \in C_{\ell'}$ and yet $V_{j_{\ell}}$ was in the partial solution set $\S$ seen at iteration $\ell'$. Therefore, it must be that
$$
d(j_{\ell'}, V_{j_{\ell}}) > (k+2) t_{j_{\ell'}} = (k+2) w_{\ell'}
$$

Putting these two inequalities together, we have shown that
$$
D_{\ell} + D_{\ell'} + \theta(j_{\ell}) > (k+2) w_{\ell'}.
$$

As $\theta(j_{\ell}) \leq w_{\ell} \leq w_{\ell'} $, this implies that
$$
\frac{D_{\ell}}{w_{\ell}} + \frac{D_{\ell'}}{w_{\ell'}} \geq \frac{D_{\ell} + D_{\ell'}}{w_{\ell'}} > \frac{ (k+2) w_{\ell'} - \theta(j_{\ell})}{w_{\ell'}} \geq \frac{ (k+2) w_{\ell'} - w_{\ell}}{w_{\ell'}} \geq \frac{ (k+2) w_{\ell'} - w_{\ell'}}{w_{\ell'}} = k+1.
$$

We have shown that, with probability one, there is some pair $\ell < \ell'$ satisfying this inequality $D_{\ell}/w_{\ell} + D_{\ell'}/w_{\ell'} > k+1$. Therefore, with probability one it holds that
\begin{equation}
\label{tr77}
\sum_{\ell=1}^{k+1} D_{\ell}/w_{\ell} > k+1.
\end{equation}

But now take expectations, observing that $\bE[D_{\ell}] = \bE[ d(j_{\ell}, \S) ] \leq t_{j_{\ell}} = w_{\ell}$. So the LHS of (\ref{tr77}) has expectation at most $k+1$. This is a contradiction.
\end{proof}

We remark that it is possible to obtain an optimal $(1, k+1)$-determinization algorithm for the SCC or homogeneous settings, but we omit this since it is very similar to Algorithm~\ref{algo:det2}.

\section{Acknowledgments}
Our sincere thanks to Brian Brubach and to the anonymous referees, for many useful suggestions and for helping to tighten the focus of the paper.

Thanks to Leonidas Tsepenekas, for pointing out an error in Theorem 14 of the journal version of the paper (appearing in the Journal of Machine Learning Research).
\bibliographystyle{abbrv}
\bibliography{kcenter-ref}

\appendix

\section{Proof of Pseudo-Theorem~\ref{thm:k-center-1.592}}
\label{proof1592}
We would like to use Lemma~\ref{prop31} to bound $\bE[d(i,\S)]$, over all possible integer values $m \geq 1$ and over all possible sequences $1 \geq u_1 \geq u_2 \geq u_3 \geq \dots \geq u_t \geq 0$. One technical obstacle here is that this is not a compact space, due to the unbounded dimension $t$ and unbounded parameter $m$. The next result removes these restrictions.

\begin{proposition}
\label{hat-t-prop1}
For any fixed integers $L, M \geq 1$, and every $j \in \C$, we have
$$
\bE[d(j,\S)] \leq 1 + \max_{\substack{m \in \{1, 2, \dots, M\} \\ 1 \geq u_1 \geq u_2 \geq \dots u_L \geq 0}} \bE_{Q} \hat R(m, u_1, u_2, \dots, u_L),
$$
where we define
{\allowdisplaybreaks
\begin{align*}
\alpha &= \prod_{\ell=1}^{L-1} (1 - \overline \qp (u_{\ell} - u_{\ell+1})) \times e^{-\overline \qp u_L}, \\
\beta  &= \prod_{\ell=1}^{L-1} (\overline{u_{\ell}} + \overline \qp u_{\ell+1}) \times \begin{cases}
(1 - u_L) & \text{if $u_L \leq \qp$} \\
e^{-\frac{u_L - \qp}{1 - \qp}}  (1 - \qp) & \text{if $u_L > \qp$} \\
\end{cases}, \\
\hat R(m, u_1, \dots, u_L)  &= \begin{cases}
(1 - \overline \qf \frac{\overline{u_1}}{m})^m   \alpha  + ( \overline \qf  (1 - \frac{\overline{u_1}}{m}) )^m \beta  & \text{if $m < M$}  \\
e^{- \overline \qf \overline{u_1}} \alpha  + \overline{\qf}^M e^{-\overline{u_1}} \beta & \text{if $m = M$}
\end{cases}.
\end{align*}
}
The expectation $\bE_Q$ is taken only over the randomness involved in $\qf, \qp$.
\end{proposition}
\begin{proof}
By Lemma~\ref{prop31},
\begin{align*}
\bE[d(i,\S) \mid \qf, \qp] & \leq 1 +  \left( 1 - \overline \qf \frac{ \overline{u_1}}{m} \right)^m  \prod_{\ell=1}^t (1 - \overline \qp (u_{\ell} - u_{\ell+1})) + \left( \overline \qf (1 - \frac{\overline{u_1}}{m}) \right)^m \prod_{\ell=1}^t (\overline{u_{\ell}} + \overline \qp u_{\ell+1}).
\end{align*}
where $u_1, \dots, u_t, m$ are defined as in Lemma~\ref{prop31}; in particular $1 \geq u_1 \geq u_2 \geq \dots \geq u_t \geq u_{t+1} = 0$ and $m \geq 1$. If we define $u_j = 0$ for all integers $j \geq t$, then
\begin{align*}
\bE[d(i,\S) \mid \qf, \qp] & \leq 1 +  \left( 1 - \overline \qf \frac{ \overline{u_1}}{m} \right)^m  \prod_{\ell=1}^\infty (1 - \overline \qp (u_{\ell} - u_{\ell+1})) + \left( \overline \qf (1 - \frac{\overline{u_1}}{m}) \right)^m \prod_{\ell=1}^\infty (\overline{u_{\ell}} + \overline \qp u_{\ell+1}).
\end{align*}

The terms corresponding to $\ell > L$ telescope so we estimate these as:
\begin{align*}
 \prod_{\ell=L}^{\infty}  (1 - \overline \qp (u_{\ell} - u_{\ell+1})) &\leq  \prod_{\ell=L}^{\infty} e^{- \overline \qp (u_{\ell} - u_{\ell+1})} = e^{-\overline \qp u_L}
\end{align*}
and
\begin{align*}
\prod_{\ell=L}^{\infty}  (\overline{u_{\ell}} + \overline \qp u_{\ell+1}  ) &\leq   (1 - u_L + \overline \qp u_{L+1})   \prod_{\ell=L+1}^{\infty} e^{-u_{\ell} + \overline \qp  u_{\ell+1}} \leq  (1 - u_L + \overline \qp u_{L+1}) e^{-u_{L+1}}.
\end{align*}

Now consider the expression $(1 - u_L + \overline \qp u_{L+1}) e^{-u_{L+1}}$ as a function of $u_{L+1}$ in the range $u_{L+1} \in [0,u_L]$. Elementary calculus shows that it satisfies the bound
$$
(1 - u_L + \overline \qp u_{L+1}) e^{-u_{L+1}} \leq
\begin{cases}
(1 - u_L) & \text{if $u_L \leq \qp$} \\
e^{-\frac{u_L - \qp}{1 - \qp}}  (1 - \qp) & \text{if $u_L > \qp$} \\
\end{cases},
$$

Thus
\begin{align*}
\bE[d(i,\S) \mid \qf, \qf] &\leq 1 +  \left( 1 - \overline \qf \frac{\overline{u_1}}{m} \right)^m  \alpha   + 
\left( \overline \qf (1 - \frac{\overline{u_1}}{m})  \right)^m \beta.
\end{align*}

If $m < M$ we are done. Otherwise, for $m \geq M$, we upper-bound the $\qf$ terms as:
\[
(1 - \overline \qf \overline{u_1}/m)^m \leq  e^{-\overline \qf \overline{u_1}}, \qquad \qquad (\overline \qf (1 - \overline{u_1}/m))^m \leq \overline{\qf}^{M} e^{-\overline{u_1}} \qedhere
\]
\end{proof}

In light of Proposition~\ref{hat-t-prop1}, we can get an upper bound on $\bE[d(i,\S)]$ by maximizing $\hat R$. We now discuss to  bound $\hat R$ for a fixed choice of $L,M$, where we select $\qf, \qp$ according to the following type of distribution:
$$
( \qf, \qp ) = \begin{cases}
(\gamma_{0,\text{f}}, 0) & \text{with probability $p$} \\
(\gamma_{1,\text{f}}, \gamma_{1,\text{p}}) & \text{with probability $1-p$}
\end{cases}.
$$

Note that we are setting $\gamma_{0,\text{p}} = 0$ here in order to keep the search space more manageable. We need to upper-bound the function $\bE_Q  \hat R(m, u_1, \dots, u_L)$, for fixed values $\gamma$ and where $m, u$ range over the compact domain $m \in \{1, \dots, M \}, 1 \geq u_1 \geq \dots \geq u_L \geq 0$. The most straightforward way to do so would be to divide $(u_1, \dots, u_L)$ into intervals of size $\epsilon$, and then calculate upper bounds on $\hat R$ within each interval and for each $m$. This process would have a runtime of $M \epsilon^{-L}$, which is too expensive.

But suppose we have fixed $u_j, \dots, u_L$, and we wish to continue to enumerate over $u_{1}, \dots, u_{j-1}$. To compute $\hat R(m, u_1, \dots, u_L)$ as a function of $m, u_{1}, \dots, u_L$, observe that we do not need to know all the  values $u_{j+1}, \dots, u_L$, but only the following four summary statistics of them:
\begin{enumerate}
\item $a_1 = e^{-\overline{\gamma_{1,\text{p}}} u_L} \prod_{\ell=j}^{L-1} (1 - \overline{\gamma_{1, \text{p}}} (u_{\ell} - u_{\ell+1}))$,
\item $a_2 = \prod_{\ell=j}^{L-1} (\overline{u_{\ell}} + \overline{\gamma_{1,\text{p}}} u_{\ell+1}) \times \begin{cases}
(1 - u_L) & \text{if $u_L \leq \gamma_{1,\text{p}}$} \\
e^{-\frac{u_L - \gamma_{1,\text{p}}}{1 - \gamma_{1,\text{p}}}} (1 - \gamma_{1,\text{p}}) & \text{if $u_L > \gamma_{1,\text{p}}$} \\
\end{cases},
$
\item $a_3 = e^{-u_L} \prod_{\ell=j}^{L-1} (1 - (u_{\ell} - u_{\ell+1}))$,
\item $a_4 = u_{j+1}$.
\end{enumerate}
We thus use a dynamic program wherein we track, for $j = L, \dots, 1$, all possible values for the tuple $(a_1, \dots, a_4)$.  Furthermore, since $\hat R$ is a monotonic function of $a_1, a_2, a_3, a_4$, we only need to store the \emph{maximal} tuples $(a_1, \dots, a_4)$.  The resulting search space has size $O(\epsilon^{-3})$. 

We wrote $C$ code to perform this computation to upper-bound $\bE_Q \hat R$ with $M = 10, \epsilon = 2^{-12}, L = 7$. This runs in about an hour on a single CPU core. With some additional tricks, we can also optimize over the parameter $p \in [0,1]$ while still keeping the stack space bounded by $O(\epsilon^{-3})$.

Note that due to the complexity of the dynamic programming algorithm,  the computations were carried out in double-precision floating point arithmetic. The rounding errors were not tracked precisely, and it would be difficult to write completely correct code to do so. We believe that these errors should be orders of magnitude below the third decimal place, and that the computed value $1.592$ is a valid upper bound. 
\end{document}